\documentclass[letterpaper, 10 pt, conference]{ieeeconf} % Comment this line out
\IEEEoverridecommandlockouts                              % This command is only
\usepackage{microtype}
\usepackage[all, nodollardollar]{onlyamsmath} % Do not remove!
\usepackage{mathtools}
\usepackage{fullpage}

\usepackage{comment}
\usepackage{setspace}
\usepackage{amsmath}
\usepackage{amssymb}
\usepackage{etoolbox}
\usepackage{color,colortbl}
\usepackage{mathtools}
\usepackage{enumerate}
\usepackage{algorithm, algpseudocode}
\usepackage{booktabs}
\usepackage{multirow}
\usepackage[normalem]{ulem}
\usepackage{pgf}
\usepackage{tikz,pgfplots}
\usetikzlibrary{trees,decorations,arrows,arrows.meta,automata,shadows,positioning,plotmarks,backgrounds,shapes,shapes.misc}
\usetikzlibrary{calc,matrix,fit,petri,decorations.pathmorphing,patterns}
\usetikzlibrary{decorations.pathreplacing,decorations.markings,shapes.geometric,calc}
\usepackage{paralist}
\usepackage{stmaryrd}
\usepackage{xspace}
\usepackage{graphicx}
\usepackage{float}
\usepackage[hang, tight]{subfigure}
\usepackage[utf8]{inputenc} 
 \usepackage[english]{babel} 
  \usepackage{csquotes} 
  \usepackage{epstopdf}

 \usepackage{array}

\usepackage{xifthen}
\usepackage{url}
\usepackage{relsize}

\newcommand\set[1]{\{ #1 \}}
\newcommand\tuple[1]{{( #1 )}}
\newcommand\tuplel[1]{{\langle #1 \rangle}}
\newcommand{\REFlem}[1]{\text{Lem.~\ref{#1}}}
\newcommand{\REFthm}[1]{\text{Thm.~\ref{#1}}}
\newcommand{\REFdef}[1]{Def.~\ref{#1}}

\newcommand{\REFsec}[1]{Sec.~\ref{#1}}
\newcommand{\REFprop}[1]{Prop.~\ref{#1}}
\newcommand{\REFfig}[1]{Fig.~\ref{#1}}

\newcommand{\REFapp}[1]{App.~\ref{#1}}

\newtheorem{theorem}{Theorem}
\newtheorem{definition}{Definition}
\newtheorem{proposition}{Proposition}

\newtheorem{lemma}{Lemma}

\newcommand{\ON}[1]{\operatorname{#1}}

\def\clap#1{\hbox to 0pt{\hss#1\hss}}

\makeatletter 
\newif\ifFIRST
\newif\ifSECOND
\let\LISTOP\relax
\newcommand{\List}[4][\;]{#3#1%
	\FIRSTtrue
	\@for\i:=#2\do{%
	\ifFIRST\LISTOP{\i}\FIRSTfalse\else,\LISTOP{\i}\fi%
	}%
	#1#4%
	\let\LISTOP\relax
}
\makeatother
\newcounter{DINGLIST}
\stepcounter{DINGLIST}
\newcommand{\markD}[3][\;\;]{\text{\ding{\the\numexpr171+\theDINGLIST}\stepcounter{DINGLIST}}#1#3}

\makeatletter

\newcommand{\propNeg}{\@ifstar\propNegStar\propNegNoStar}
\newcommand{\propNegStar}[1]{\ensuremath{\left(\propNegNoStar{#1}\right)}}
\newcommand{\propNegNoStar}[2][\cdot]{\ensuremath{\neg\ifthenelse{\isempty{#2}}{#1}{#2}}}

\newcommand{\propConj}{\@ifstar\propConjStar\propConjNoStar}
\newcommand{\propConjStar}[2]{\ensuremath{\left(\propConjNoStar{#1}{#2}\right)}}
\newcommand{\propConjNoStar}[3][\cdot]{\ensuremath{\ifthenelse{\isempty{#2}}{#1}{#2}\wedge\ifthenelse{\isempty{#3}}{#1}{#3}}}

\newcommand{\propDisj}{\@ifstar\propDisjStar\propDisjNoStar}
\newcommand{\propDisjStar}[2]{\ensuremath{\left(\propDisjNoStar{#1}{#2}\right)}}
\newcommand{\propDisjNoStar}[3][\cdot]{\ensuremath{\ifthenelse{\isempty{#2}}{#1}{#2}\vee\ifthenelse{\isempty{#3}}{#1}{#3}}}

\newcommand{\propImp}{\@ifstar\propImpStar\propImpNoStar}
\newcommand{\propImpStar}[2]{\ensuremath{\left(\propImpNoStar{#1}{#2}\right)}}
\newcommand{\propImpNoStar}[3][\cdot]{\ensuremath{\ifthenelse{\isempty{#2}}{#1}{#2}\Rightarrow\ifthenelse{\isempty{#3}}{#1}{#3}}}

\newcommand{\propAequ}{\@ifstar\propAequStar\propAequNoStar}
\newcommand{\propAequStar}[2]{\ensuremath{\left(\propAequNoStar{#1}{#2}\right)}}
\newcommand{\propAequNoStar}[3][\cdot]{\ensuremath{\ifthenelse{\isempty{#2}}{#1}{#2}\Leftrightarrow\ifthenelse{\isempty{#3}}{#1}{#3}}}

\newcommand{\AllQ}{\@ifstar\AllQStar\AllQNoStar}
\newcommand{\AllQStar}[3][\;]{\ensuremath{\left(\forall #2#1.#1#3\right)}}
\newcommand{\AllQNoStar}[3][\;]{\ensuremath{\forall #2#1.#1#3}}
\newcommand{\AllQu}{\@ifstar\AllQuStar\AllQuNoStar}
\newcommand{\AllQuStar}[3][\;]{\ensuremath{\left(\forall^{\infty} #2#1.#1#3\right)}}
\newcommand{\AllQuNoStar}[3][\;]{\ensuremath{\forall^{\infty} #2#1.#1#3}}

\newcommand{\ExQ}{\@ifstar\ExQStar\ExQNoStar}
\newcommand{\ExQStar}[3][\;]{\ensuremath{\left(\exists #2#1.#1#3\right)}}
\newcommand{\ExQNoStar}[3][\;]{\ensuremath{\exists #2#1.#1#3}}

\newcommand{\NExQ}{\@ifstar\NExQStar\NExQNoStar}
\newcommand{\NExQStar}[3][\;]{\ensuremath{\left(\nexists #2#1.#1#3\right)}}
\newcommand{\NExQNoStar}[3][\;]{\ensuremath{\nexists #2#1.#1#3}}

\newcommand{\UniqueQ}{\@ifstar\UniqueQStar\UniqueQNoStar}
\newcommand{\UniqueQStar}[3][\;]{\ensuremath{\left(\exists! #2#1.#1#3\right)}}
\newcommand{\UniqueQNoStar}[3][\;]{\ensuremath{\exists! #2#1.#1#3}}

  \newlength{\SFS@HEIGHT}
  \newlength{\SFS@WIDTH}
  \newcommand{\SplitX}[2]{
	    \settoheight{\SFS@HEIGHT}{$#2$}
	    \settowidth{\SFS@WIDTH}{$#2$}
	    \mbox{\begin{tikzpicture}[baseline=(current bounding box.center)]
	    \node[] (E) at (0,0) {$#1$};
	    \node[inner sep=0pt] (F) at ($(E.south west)+(1ex,-1ex)+(3ex+.5\SFS@WIDTH,-\SFS@HEIGHT)$) {$#2$};
	    \node[] (E) at (0,0) {\phantom{$#1$}};
	    \draw[fill] ($(E.east)+(1ex,0ex)$) circle (.2ex);
	    \draw[-] ($(E.east)+(1ex,0ex)$) -- ($(E.south east)+(1ex,-0.5ex)$) -- ($(E.south west)+(1ex,-0.5ex)$) -- ($(E.south west)+(1ex,-1ex)-(0,\SFS@HEIGHT)$) -- ($(E.south west)+(2.5ex,-1ex)-(0,\SFS@HEIGHT)$);
	    \draw[fill] ($(E.south west)+(2.5ex,-1ex)-(0,\SFS@HEIGHT)$) circle (.2ex);
	    \end{tikzpicture}}}
  \newcommand{\SplitS}[2]{
	    \settoheight{\SFS@HEIGHT}{$#2$}
	    \settowidth{\SFS@WIDTH}{$#2$}
	    \mbox{\begin{tikzpicture}[baseline=(current bounding box.center)]
	    \node[] (E) at (0,0) {$#1$};
	    \node[inner sep=0pt] (F) at ($(E.south west)+(1ex,0.5ex)+(3ex+.5\SFS@WIDTH,-\SFS@HEIGHT)$) {$#2$};
	    \end{tikzpicture}}}

\newcommand{\semantics}[1]{\langle\![#1]\!\rangle}
\newcommand{\Set}[2][]{\List[#1]{#2}{\left\{}{\right\}}}
\newcommand{\VSet}[2][]{\let\LISTOP\val\List[#1]{#2}{\{}{\}}}

\newcommand{\Tuple}[2][]{\List[#1]{#2}{(}{)}}
\newcommand{\VTuple}[2][]{\let\LISTOP\val\List[#1]{#2}{(}{)}}
\newcommand{\UNION}{\@ifstar\UNIONStar\UNIONNoStar}
\newcommand{\UNIONStar}[2]{\ensuremath{\left(\UNIONNoStar{#1}{#2}\right)}}
\newcommand{\UNIONNoStar}[2]{\ensuremath{\ifthenelse{\isempty{#1}}{\cdot}{#1}\cup\ifthenelse{\isempty{#2}}{\cdot}{#2}}}

\newcommand{\UNIOND}{\@ifstar\UNIONDStar\UNIONDNoStar}
\newcommand{\UNIONDStar}[2]{\ensuremath{\left(\UNIONDNoStar{#1}{#2}\right)}}
\newcommand{\UNIONDNoStar}[2]{\ensuremath{\ifthenelse{\isempty{#1}}{\cdot}{#1}\uplus\ifthenelse{\isempty{#2}}{\cdot}{#2}}}

\newcommand{\SETMINUS}{\@ifstar\SETMINUSStar\SETMINUSNoStar}
\newcommand{\SETMINUSStar}[2]{\ensuremath{\left(\SETMINUSNoStar{#1}{#2}\right)}}
\newcommand{\SETMINUSNoStar}[2]{\ensuremath{\ifthenelse{\isempty{#1}}{\cdot}{#1}\setminus\ifthenelse{\isempty{#2}}{\cdot}{#2}}}

\newcommand{\INTERSECT}{\@ifstar\INTERSECTStar\INTERSECTNoStar}
\newcommand{\INTERSECTStar}[2]{\ensuremath{\left(\INTERSECTNoStar{#1}{#2}\right)}}
\newcommand{\INTERSECTNoStar}[2]{\ensuremath{\ifthenelse{\isempty{#1}}{\cdot}{#1}\cap\ifthenelse{\isempty{#2}}{\cdot}{#2}}}

\newcommand{\CARTPROD}{\@ifstar\CARTPRODStar\CARTPRODNoStar}
\newcommand{\CARTPRODStar}[2]{\ensuremath{\left(\CARTPRODNoStar{#1}{#2}\right)}}
\newcommand{\CARTPRODNoStar}[2]{\ensuremath{\ifthenelse{\isempty{#1}}{\cdot}{#1}\times\ifthenelse{\isempty{#2}}{\cdot}{#2}}}

\newcommand{\FINCOUNT}{\@ifstar\FinCountStar\FinCountNoStar}
\newcommand{\FinCountStar}[1]{\ensuremath{\#(\ifthenelse{\isempty{#1}}{\cdot}{#1})}}
\newcommand{\FinCountNoStar}[1]{\ensuremath{\#\left(\ifthenelse{\isempty{#1}}{\cdot}{#1}\right)}}

\makeatother

\newcommand{\real}[1]{\ifstrempty{#1}{\mathbb{R}}{\mathbb{R}^{#1}}}
\newcommand{\Z}{\mathbb{Z}}
\newcommand{\N}{\mathbb{N}}

\newcommand{\fun}{\ensuremath{\ON{\rightarrow}}}
\newcommand{\setfun}{\ensuremath{\ON{\rightrightarrows}}}
\newcommand{\SetComp}[3][]{\{#1#2#1\mid#1#3#1\}}

\newcommand{\dom}[1]{\ensuremath{\mathrm{dom}(#1)}}

\newcommand{\hyint}[1]{\ensuremath{\llbracket #1 \,\rrbracket}}
\newcommand{\frr}[2]{\preccurlyeq_{#1}^{#2}}

\newcommand{\frrE}[2]{\cong_{#1}^{#2}}

\newcommand{\AP}{\ensuremath{\mathtt{AP}}}
\newcommand{\I}{\ensuremath{\Upsilon}}
\renewcommand{\O}{\ensuremath{\Lambda}}
\newcommand{\Ia}{\ensuremath{\widehat{\I}}}
\newcommand{\Oa}{\ensuremath{\widehat{\O}}}

\newcommand{\Ie}{\ensuremath{\I^\star}}
\newcommand{\Oe}{\ensuremath{\O^\star}}

\renewcommand{\i}{\ensuremath{\mu}}
\renewcommand{\o}{\ensuremath{\lambda}}

\newcommand{\Enab}{\ensuremath{\ON{Enab}}}
\newcommand{\Omap}[1]{\ensuremath{P_\O\ifthenelse{\isempty{#1}}{}{(#1)}}}
\newcommand{\Imap}[1]{\ensuremath{P_\I\ifthenelse{\isempty{#1}}{}{(#1)}}}
\newcommand{\Omapa}[1]{\ensuremath{\widehat{P}_{\O}\ifthenelse{\isempty{#1}}{}{(#1)}}}
\newcommand{\Imapa}[1]{\ensuremath{\widehat{P}_{\I}\ifthenelse{\isempty{#1}}{}{(#1)}}}
\newcommand{\Map}[1]{\ensuremath{P\ifthenelse{\isempty{#1}}{}{(#1)}}}
\newcommand{\Mapa}[1]{\ensuremath{\widehat{P}\ifthenelse{\isempty{#1}}{}{(#1)}}}
\newcommand{\Mapai}[1]{\ensuremath{\widehat{P}^{-1}\ifthenelse{\isempty{#1}}{}{(#1)}}}

\newcommand{\Omape}[1]{\ensuremath{P^\star_{\Oe}\ifthenelse{\isempty{#1}}{}{(#1)}}}
\newcommand{\Imape}[1]{\ensuremath{P^\star_{\Ie}\ifthenelse{\isempty{#1}}{}{(#1)}}}
\newcommand{\Mape}[1]{\ensuremath{P^\star\ifthenelse{\isempty{#1}}{}{(#1)}}}
\newcommand{\Mapei}[1]{\ensuremath{{P^\star}^{-1}\ifthenelse{\isempty{#1}}{}{(#1)}}}

\newcommand{\Mapc}[1]{\ensuremath{\check{P}\ifthenelse{\isempty{#1}}{}{(#1)}}}
\newcommand{\Omapc}[1]{\ensuremath{\check{P}_{\O}\ifthenelse{\isempty{#1}}{}{(#1)}}}
\newcommand{\Imapc}[1]{\ensuremath{\check{P}_{\I}\ifthenelse{\isempty{#1}}{}{(#1)}}}

\newcommand{\Sa}{\ensuremath{\widehat{S}}}

\newcommand{\Xa}{\ensuremath{\widehat{X}}}
\newcommand{\Xoa}{\ensuremath{\widehat{X}_0}}
\newcommand{\Xao}{\ensuremath{\widehat{X}_0}} % remove
\newcommand{\Ua}{\ensuremath{\widehat{U}}}
\newcommand{\Ya}{\ensuremath{\widehat{Y}}}
\newcommand{\Fa}{\ensuremath{\widehat{F}}}
\newcommand{\Ha}{\ensuremath{\widehat{H}}}
\newcommand{\pa}{\ensuremath{\widehat{\psi}}}
\newcommand{\Qa}{\ensuremath{\widehat{Q}}}
\newcommand{\qa}{\ensuremath{\widehat{q}}}

\newcommand{\p}{\ensuremath{\psi}}

\newcommand{\Q}{\ensuremath{\mathcal{Q}}}

\newcommand{\xa}{\ensuremath{\widehat{x}}}
\newcommand{\ya}{\ensuremath{\widehat{y}}}

\newcommand{\sigmaa}{\ensuremath{\widehat{\sigma}}}

\newcommand{\pia}{\ensuremath{\widehat{\pi}}}
\newcommand{\ua}{\ensuremath{\widehat{u}}}

\newcommand{\dap}{\ensuremath{\widehat{\delta}_{\times}}}

\newcommand{\Xo}{\ensuremath{X_0}}

\newcommand{\C}{\ensuremath{\mathcal{C}}}
\newcommand{\Ca}{\ensuremath{\widehat{\mathcal{C}}}}
\newcommand{\Paths}[1]{\ensuremath{\ON{Paths}(#1)}}
\newcommand{\CPaths}[1]{\ensuremath{\ON{CPaths}(#1)}}

\newcommand{\EPaths}[1]{\ensuremath{\ON{EPaths}(#1)}}

\newcommand{\Obs}[1]{\ensuremath{\ON{Obs}\ifthenelse{\isempty{#1}}{}{(#1)}}}
\newcommand{\Ext}[1]{\ensuremath{\ON{Ext}\ifthenelse{\isempty{#1}}{}{(#1)}}}
\newcommand{\Extn}[1]{\ensuremath{\ON{Ext}_{#1}}}
\renewcommand{\State}[1]{\ensuremath{\ON{State}\ifthenelse{\isempty{#1}}{}{(#1)}}}

\newcommand{\Prefs}[1]{\ensuremath{\ON{Prefs}(#1)}}

\newcommand{\CPrefs}[1]{\ensuremath{\ON{CPrefs}(#1)}}

\newcommand{\Last}[1]{\ensuremath{\ON{Last}\ifthenelse{\isempty{#1}}{}{(#1)}}}
\newcommand{\ELast}[1]{\ensuremath{\ON{ELast}\ifthenelse{\isempty{#1}}{}{(#1)}}}
\newcommand{\LastS}[1]{\ensuremath{\ON{LastX}\ifthenelse{\isempty{#1}}{}{(#1)}}}
\newcommand{\LastSn}[2]{\ensuremath{\ON{LastX}_{#1}\ifthenelse{\isempty{#2}}{}{(#2)}}}
\newcommand{\EHistn}[2]{\ensuremath{\ON{EHist}_{#1}\ifthenelse{\isempty{#2}}{}{(#2)}}}
\newcommand{\EHist}[1]{\ensuremath{\ON{EHist}\ifthenelse{\isempty{#1}}{}{(#1)}}}
\newcommand{\Hist}[1]{\ensuremath{\ON{Hist}\ifthenelse{\isempty{#1}}{}{(#1)}}}
\newcommand{\Histn}[2]{\ensuremath{\ON{Hist}_{#1}\ifthenelse{\isempty{#2}}{}{(#2)}}}

\newcommand{\EPrefs}[1]{\ensuremath{\ON{EPrefs}(#1)}}
\newcommand{\EPrefsk}[2]{\ensuremath{\ON{EPrefs}_{#2}(#1)}}

\newcommand{\WIN}{\ensuremath{\mathcal{W}}}

\newcommand{\Aut}{\ensuremath{\mathcal{A}}}

\newcommand{\Auts}{\ensuremath{\mathcal{A}_{\psi}}}

\newcommand{\Autap}{\ensuremath{\widehat{\mathcal{A}}_{\times}}}

\newcommand{\acc}{\ensuremath{\mathcal{F}}}
\newcommand{\acca}{\ensuremath{\widehat{\mathcal{F}}}}

\newcommand{\Lang}{\ensuremath{\mathcal{L}}}

  \newcommand{\fz}[1]{\ensuremath{f^{0}\ifthenelse{\isempty{#1}}{}{\Tuple{#1}}}} 
  \newcommand{\fo}[1]{\ensuremath{f^{1}\ifthenelse{\isempty{#1}}{}{\Tuple{#1}}}}

\title{%\LARGE \bf
Abstraction-Based Output-Feedback Control with State-Based Specifications} 

\author{
Anne-Kathrin Schmuck \and Mehrdad Zareian	
\thanks{
A.-K. Schmuck and M. Zareian are with the Max Planck Institute for Software Systems (MPI-SWS), Kaiserslautern, Germany.
Email: {\tt\small \{akschmuck,mzareian\}@mpi-sws.org}
}
}

\begin{document}

%---------- Title ----------
\maketitle
\thispagestyle{empty}
%\pagestyle{empty}

%---------- Abstract ----------
\begin{abstract}
We consider abstraction-based design of \emph{output-feedback} controllers for non-linear dynamical systems against specifications over \emph{state-based} predicates in linear-time temporal logic (LTL).
In this context, our contribution is two-fold: (I) we generalize feedback-refinement relations for abstraction-based output-feedback control to systems with arbitrary predicate and observation maps, and (II) we introduce a new algorithm for the synthesis of abstract \emph{output-feedback} controllers w.r.t.\ LTL specifications over unobservable \emph{state-based} predicates.

Our abstraction-based output-feedback controller synthesis algorithm consists of two steps. First, we compute a finite state abstraction of the original system using existing techniques. This process typically leads to an abstract system with non-deterministic predicate and observation maps which are not necessarily related to each other. Second, we introduce an algorithm to compute an output-feedback controller for such abstract systems. Our algorithm is inspired by reactive synthesis under partial observation and utilizes \emph{bounded synthesis}.
\end{abstract}

\section{Introduction}

Abstraction-based control design (ABCD) is a well known technique to synthesize
correct-by-design control software for cyber-physical systems. In particular, ABCD allows to consider continuous-state dynamical systems in combination with discrete, temporal control objectives and computes controllers almost fully automatically. 

ABCD comes in various flavors implemented in different tools and applicable to different types of dynamical systems and classes of temporal specifications, e.g. \cite{ReissigWeberRungger_2017_FRR,HsuMMS18,pFaces,nilsson2017augmented,li2018rocs}. However, almost all works on ABCD rely on the fact that the state of the system is observable. This can be a very restrictive assumption in practice. Recently, this requirement was relaxed and \emph{abstraction-based output-feedback control design} (ABoCD) was considered and follows mostly two different approaches.

In the first approach, classical observers for the original dynamical system are computed and incorporated into the abstraction process \cite{Mickelin2014synthesis,haesaert2015correct,Pola2019symbolic,ApazaGirard2020symbolic}. This handles the complexity of output-feedback control in the pre-abstraction phase and thereby requires particular properties of the underlying dynamics and observation maps to allow for observer design. 
 The second approach does not assume any \enquote{niceties} of the underlying dynamics or observation maps and moves the complexity of output-feedback control to the abstract layer by considering a partial-observation game for discrete output-feedback control \cite{MizoguchiUshio2018deadlockfree,MajumdarNS2020,khaled2020outputfeedback}.
 
 Within this paper, we follow the second approach. Here, existing works are either limited to control problems where specifications are defined over observables \cite{MajumdarNS2020}, only safety specifications over abstract states are considered \cite{MizoguchiUshio2018deadlockfree}, or it is required that the computed abstraction is detectable, i.e., after a finite number of steps the true abstract state can be determined \cite{khaled2020outputfeedback}.

\begin{figure}
	\centering
	\includegraphics[width = 0.7 \columnwidth]{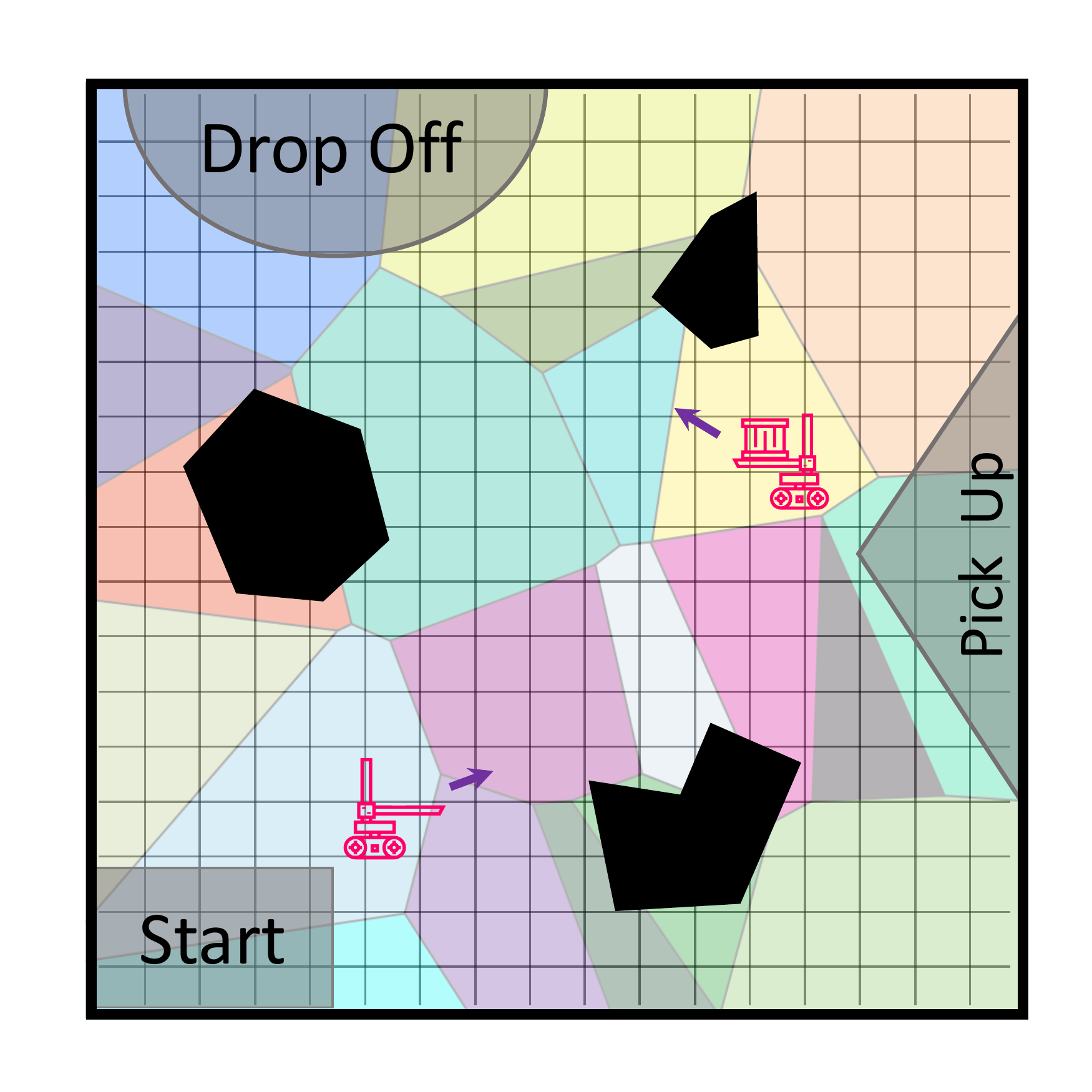}
	\vspace{-0.6cm}
	\caption{Robot motion planning example with restricted position measurement (only colored tiles are observable) and un-observable predicates (black and gray) induced by the specification which requires to alternate between \enquote{pick up} and \enquote{drop off} while avoiding obstacles (black). Abstraction results in the grided state space, where each cell is an abstract state with non-deterministic observation and predicate map (each cell might intersect with multiple tiles or predicates). }
	\vspace{-0.5cm}
	\label{fig:EX1}
\end{figure}

Within this paper, we relax these assumptions on ABoCD as exemplified in the robotic motion planning example depicted in \REFfig{fig:EX1}. Here, a mobile robot (with non-linear disturbed dynamics) can only sense its location by detecting the color of the tile it is currently moving over. This can for example be realized by a downward pointing sensor that detects different colored lines on the floor which indicate tile boundaries. Using this restricted (but very cheap) position measurement, the robot should be controlled such that it alternates between the \enquote{pick up} and \enquote{drop off} location, while avoiding the black obstacles. We see that specification predicates (i.e., the black and gray regions) are not expressible in terms of the (restricted) observations and hence are not observable. 

In order to synthesize an output-feedback controller for such a system, a straightforward approach is to first employ a uniform grid-based abstraction technique, as e.g.\ in \texttt{SCOTS} \cite{Scots}, to generate a finite abstraction. For the example in \REFfig{fig:EX1}, the abstract system would have one abstract state per grid cell. When computing this abstraction, we additionally need to transform the predicate- and observation maps. As we see from \REFfig{fig:EX1}, typically multiple predicates or tiles intersect with a single boxed grid cell. This leads to non-deterministic predicate and observation maps on the abstraction.

This example demonstrates that ABoCD with non-observable predicates requires to
 \begin{compactenum}[(I)]
 \item extend the notion of \emph{feedback-refinement relations} (FRR) \cite{ReissigWeberRungger_2017_FRR} to systems with non-deterministic predicate and observation maps, and 
 \item to develop an algorithm that synthesizes output-feedback controllers for such systems.
\end{compactenum}

Within this paper, we tackle challenge (I) in \REFsec{sec:soundabs} where we define  \emph{extended feedback-refinement relations} (eFRR) and sound abstract specifications. As our first contribution, this provides a new framework for sound ABoCD in the presence of unobservable predicates.

Afterwards, we address challenge (II) in two steps. We first show in \REFsec{sec:Abst:A} that employing a standard grid-based abstraction technique, as in \texttt{SCOTS} \cite{Scots} with the obvious transformation of predicate- and observation maps, yields a sound finite abstraction which allows for an eFRR to the original system.  
As our second contribution we then show in \REFsec{sec:Abst:B} how an abstract output-feedback controller can be synthesized for this finite abstraction which has possibly non-deterministic predicate and observation maps. We emphasize that this algorithm does not require any pre-processing of the predicate map. I.e., we do not need to smartly \enquote{expand} or \enquote{shrink} predicates to render predicate maps deterministic, as required for state-based ABCD in \texttt{SCOTS}.

\begin{figure}[t]
	\centering
	\includegraphics[width = 1\columnwidth]{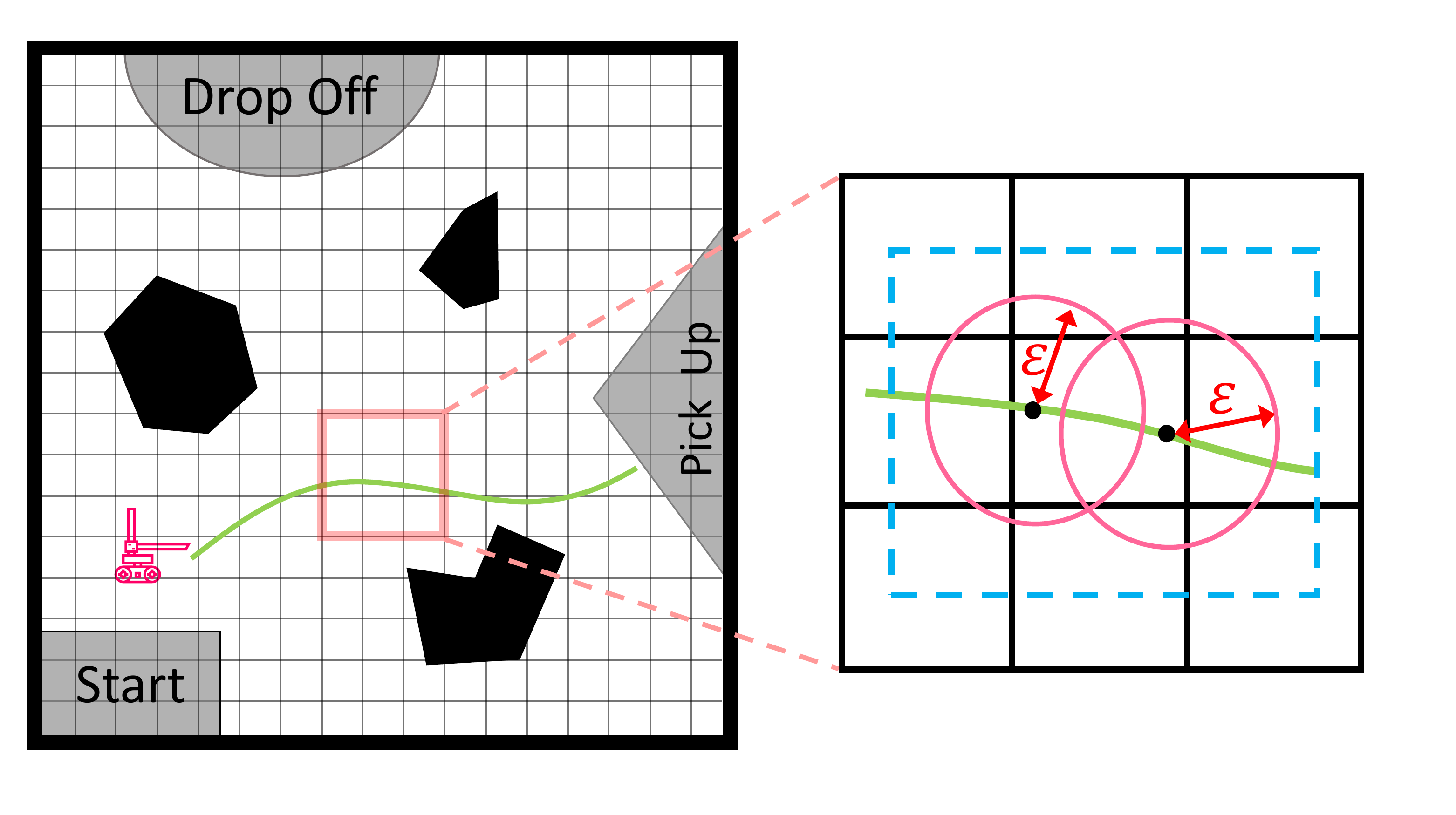}
	\vspace{-0.9cm}
	\caption{Another application scenario for our ABoCD technique. Here the position of the robot is measured with some error bounded by $\varepsilon$. After grid-based abstraction this results in a non-deterministic output-map (indicated by the blue dashed line on the right) even if the abstract state and output spaces coincide.}
	\vspace{-0.5cm}
	\label{fig:EX2}
\end{figure}

Our resulting ability to synthesize output-feedback controllers for finite systems with non-deterministic output maps, also allows us to tackle a slightly different abstraction-based control problem exemplified in \REFfig{fig:EX2}, which is more closely related to the problems studied in the first approach to ABoCD (e.g., in \cite{ApazaGirard2020symbolic}). Here, the position of the robot can be measured with some error $\varepsilon$. When abstracting this system using a gird-based approach, the resulting abstract observation map also becomes non-deterministic, as shown in the right side of \REFfig{fig:EX2} even if we choose the abstract output space identical to the abstract state space.

%!TEX root = main.tex

\section{Preliminaries}\label{sec:prelim}

% \subsection{Preliminary Definitions}

 \smallskip 
\noindent\textbf{Notation.}
We use the symbols $\N$, $\Z$,  $\real{}$, and $\real{}_{>0}$ 
to denote the sets of natural numbers, integers, reals, and positive reals, respectively. 
Given $a,b\in\real{}$ s.t.\ $a\leq b$, we denote by $[a,b]$ a closed interval and define $[a;b]=[a,b]\cap \Z$ as its integer counterpart. 
For a set $W$, we write $W^*$ and $W^\omega$ for the sets of finite and infinite sequences over $W$, respectively, and define $W^\infty:=W^*\cup W^\omega$. 
 
For $w\in W^*$, we write $|w|$ for the length of $w$ and $\varepsilon$ for the empty string with $|\varepsilon|=0$; the length of $w\in W^\omega$ is $\infty$. 
We define $\dom{w} = \Set{0,\ldots, |w|-1}$ if $w\in W^*$, and $\dom{w} = \N$ if $w\in W^\omega$. 
For $k\in \dom{w}$ we write $w(k)$ for the $k$-th symbol of $w$
and $w|_{[0;k]}$ for the restriction of $w$ to the domain $[0;k]$. 
Given two sets $A$ and $B$, $f:A\setfun B$ and $f:A\fun B$ denote a set-valued and ordinary map, respectively. %$f$ is called \emph{strict} if $f(a)\neq\emptyset$ for all $a\in A$.  
The inverse mapping $f^{-1}:B\setfun A$ is defined via its respective binary relation: $f^{-1}(b)=\SetComp{a\in A}{b\in f(a)}$. 
By slightly abusing notation, we lift maps to subsets of their domain in the usual way, i.e., for a set-valued map $f:A\setfun B$ and $\alpha\subseteq A$ we have $f(\alpha)=\SetComp{b}{\ExQ{a\in\alpha}{b\in f(a)}}$, and similarly for ordinary maps. For any set $A$ we denote the identity function over $A$ by $\iota$ i.e., $\iota(a)=a$ for all $a\in A$.

 \smallskip 
\noindent\textbf{Systems.}
A \emph{system} $S=(X,\Xo,U,F,Y,H)$ consists of a state space $X$, 
a set of initial states $\Xo\subseteq X$, 
an input space $U$, 
a %strict set-valued 
transition function $F:X\times U\setfun X$,
an output space $Y$, and 
an output function $H:X\setfun Y$.
The only restriction we impose on such systems is $H(x)\neq\emptyset$ for all $x\in X$, i.e., we require that $Y$ is a cover of $X$. The system $S$ is called \emph{finite} if $X$, $U$ and $Y$ are finite sets. 

We lift the functions $F$ and $H$ to sets of states $A\subseteq X$ and $B\subseteq U$ s.t.\
% \begin{align}
 $F(A,B):=\bigcup_{a\in A}\bigcup_{b\in B} F(a,b)$, and $H(A)=\bigcup_{a\in A} H(a)$. If not explicitly defined otherwise, we apply this \enquote{lifting} of maps from single elements to sets via their \emph{union} to all maps defined in this paper. 
 
 Given a state $x\in X$ we define the set of \emph{enabled inputs} in $x$ as
$\Enab_{S}(x):=\SetComp{u\in U}{F(x,u)\neq\emptyset}$. We lift this map to sets of states in a slightly unusual fashion by taking \emph{intersection} rather then \emph{union}. That is, given a set $A\subseteq X$ we define
% \begin{align}
 $\Enab_{S}(A):=\bigcap_{a\in A} \Enab_{S}(a)$.
% To simplify notation, we assume that $H$ respects $\Xo$, that is, if $H^{-1}(y)\cap \Xo\neq\emptyset$ we have $H^{-1}(y)\subseteq \Xo$. 

% \end{align}

% If $Y$ and $H$ are not specified we call the tuple $(X,\Xo,U,F)=:M$ a finite state machine.

% and state propositions $\AP_S$. We interpre specifications with input and output propositions only over finite input/output systems, and assume that $U:=2^{\AP_I}$ and $Y:=2^{\AP_I}$. 
% For state proposition we define the set $P:=2^{\AP_S}$. Given a system $S$ we interpret state propositions over $S$ by introducing a map $\Omap{}:X\setfun P$. This allows us to interpret state-propositions w.r.t.\ both \emph{finite} and \emph{infinite state} systems. We remark that, in principle, this is also possibe for input and output propositions by adding appropriate maps. We refrain from doing so in this submission for the sake of notational simplicity. 
% 
% In principle our framework allows to consider all combinations of propositions in $\psi$. However, for the sake of notational simplicity, we only consider two instances, input/output specifications and input/state specifications.

% Let $D\in 2^{\set{I,O,S}}$. Then we write $\psi_D$ to explicate that $\psi$ is defined over $\AP$ divided into $|D|$ subsets $\AP_i$ with $i\in D$, i.e., $\AP=\dot{\bigcup}_{i\in D}\AP_i$. If $S\in D$ we call $\psi_D$ a \emph{state-based specification}.

\smallskip
\noindent\textbf{Trace Semantics.}
A (maximal) \emph{path} of $S$ is a sequence $\pi=x_0u_0x_1u_1\hdots$ such that $x_0\in X_0$, for all $k\in\dom{\pi}$ we have $x_{k}\in F(x_{k-1},u_{k-1})$, and if $\dom{\pi}=k<\infty$ we have $F(x_{k},u)=\emptyset$ for all $u\in U$. 
The set of all paths over $S$ is denoted by $\ON{Paths}(S)$. 
The \emph{prefix up to $x_n$} of a path $\pi$ over $S$ is denoted by $\pi|_{[0;n]}=x_0u_0x_1u_1\hdots x_n$ with length $|\pi|_{[0;n]}|=n+1$ and last element $\ON{Last}(\pi_{[0;n]})=x_n$. 
The set of all such prefixes is denoted by $\ON{Prefs}(S)$. 

Given a path $\pi$ an \emph{external} sequence $\sigma=y_0u_0y_1u_1\hdots$ is generated by $\pi$ if $y_k\in H(x_k)$ for all $k\in\dom{k}$, denoted by $\sigma\in\Ext{\pi}$. 
The set of all external sequences of a system $S$ is defined by  $\ON{EPaths}(S):=\Ext{\Paths{S}}$ with its prefix set $\ON{EPrefs}(S):=\Ext{\Prefs{S}}$.

% All external sequences generated by $\pi$ are collected in the set $\ON{Ext}(\pi)$. 
% The set of all external sequences over $S$ is denoted by $\EPaths{S}$ and we define $\EPrefs{S}:=\Ext{\Prefs{S}}$.
% 
% The set $\EPaths{S}$ is called \emph{topologically closed} (or \emph{closed} for short) if for any infinite sequence $\sigma=y_0u_0y_1u_1\hdots\in Y(UY)^\omega$,  
% whenever $\sigma_{[0;k]}\in\EPrefs{S}$ for all $k\in\N$ it holds that $\sigma\in\EPaths{S}$. We say that $S$ has \emph{closed external behavior} if $\EPaths{S}$ is closed. We further note that the set $\XPaths{S}$ is always closed for systems $S$ as defined in this paper (see, e.g., \cite{Willems} for details).

We lift the map $\Last{}$ to external sequences and write 
$x\in\LastSn{S}{\sigma}$ if there exists $\pi\in\Prefs{S}$ s.t.\ $\sigma\in\Ext{\pi}$ and $x=\Last{\pi}$.
For a state $x\in X$ we define all prefixes of $S$ that reach $x$ as $\Histn{S}{x}=\SetComp{\pi\in\Prefs{S}}{\Last{\pi}=x}$ and all external sequences generated by such prefixes as $\EHistn{S}{x}=\SetComp{\sigma\in\EPrefs{S}}{x\in\LastSn{S}{\sigma}}$. 

\smallskip 
\noindent\textbf{Control Strategies.}
We define \emph{output-feedback control strategies} as functions $\C:\EPrefs{S}\fun U$. We say that $\C$ is \emph{feedback-composable} with $S$ iff we can iteratively construct their external closed-loop behavior as follows. First, we define $\EPrefsk{S,\C}{0}:=H(\Xo)$. Then, for all $k\in\N$ we require that $\nu\in\EPrefsk{S,\C}{k}$ implies that $\C(\nu)$ is defined and $\C(\nu)\in\Enab(\LastSn{S}{\nu})$. Further, we define $\EPrefsk{S,\C}{k+1}:=\SetComp{\nu u y\in\EPrefs{S}}{\nu\in\EPrefsk{S,\C}{k},u=\C(\nu)}$. We have $\EPrefs{S,\C}:=\bigcup_{k\in\N}\EPrefsk{S,\C}{k}$ and define the set of infinite external closed-loop sequences of $S$ under $\C$ as the set $\EPaths{S,\C}\subseteq Y(UY)^\omega$ s.t.\ $\sigma\in\EPaths{S,\C}$ iff $\sigma|_{[0;k]}\in\EPrefs{S,\C}$ for all $k\in\N$.
We further define $\CPrefs{S,\C}:=\Extn{S}^{-1}(\EPrefs{S,\C})$ and $\CPaths{S,\C}:=\Extn{S}^{-1}(\EPaths{S,\C})$.

\smallskip 
\noindent\textbf{Specifications.}
We consider $\omega$-regular specifications over 
a finite set of atomic (boolean) propositions $\AP$ which are given by a formula $\psi$ in \emph{linear temporal logic} (LTL). %input and output propositions $\AP_I$ and $\AP_O$. 
We omit the standard definitions of $\omega$-regular languages and LTL (see, e.g., \cite{Thomas90,Thomas95}). % and \cite{citation_on_LTL}).
We assume that the set of atomic propositions $\AP$ can be divided into
input propositions $\AP_I$ and output propositions $\AP_O$ defining the finite sets 
$\I:=2^{\AP_I}$ and $\O:=2^{\AP_O}$ of predicates, which collect all possible sets of currently true propositions.

We interpret a specification $\psi$ on a system $S$ with the help of two predicate maps $\Imap{}:U\setfun \I$ and $\Omap{}:X\setfun \O$. Given a predicate sequence $\nu=\lambda_0\mu_0\lambda_1\mu_1\hdots$ and a path $\pi=x_0u_0x_1u_1\hdots$ of $S$, we say that $\nu$ is generated by $\pi$, written $\nu\in\Map{\pi}$, iff for all $k\in\dom{\pi}$ holds that $\mu_k\in\Imap{u_k}$ and $\lambda_k\in\Omap{x_k}$.

% We adopt the usual convention to interpret a specification $\psi$ as a $\omega$-regular language $\semantics{\psi}\subseteq \O(\I\O)^\omega$ of desired \emph{infinite} input/output sequences. I.e., we implicitly require that the closed-loop system must never block. %Using this convention implies that we are only looking for controllers which render the closed-loop system \emph{non-blocking}. 

\smallskip 
\noindent\textbf{Control Problem.}
% We consider two szenarios: input/output-based and state-based specifications. For input/output-based specifications we assume a \emph{finite input/output system} $S$ and a partition of $\AP$ into input and output propositions $\AP_I$ and $\AP_O$ s.t.\ $\AP=\AP_I\dot{\cup}\AP_O$, $U=\twoup{\AP_I}$ and $Y=\twoup{\AP_O}$. For state-based specifications, we assume a \emph{finite state system} $S$ s.t.\ $X=\twoup{\AP}$.
% Then an $\omega$-regular input/output-based (resp. state-based ) specification $\psi$ (resp. $\psi^\dagger$) can be written as a language $\semantics{\psi}\subseteq Y(UY)^\omega$ (resp. $\semantics{\psi^\dagger}\subseteq X^\omega$) of desired sequences.
% 
Given a system $S$, a specification $\psi$ and an interpretation of $\psi$ on $S$ via $\Imap{}$ and $\Omap{}$ defining the map $\Map{}$, the \emph{output-feedback control problem}, written $\tuplel{S, \psi,\Map{}}$, asks to find an \emph{output-feedback} control strategy
$\C$ which is feedback-composable with $S$ and all closed-loop paths of $S$ under $\C$ fulfill the specification.

To formalize this further, we adopt the usual convention to interpret a specification $\psi$ as a $\omega$-regular language $\semantics{\psi}\subseteq \O(\I\O)^\omega$ of desired \emph{infinite} predicate sequences.
With this, we can define set $\WIN(S,\psi,\Map{})$ of \emph{sound} output-feedback control strategies s.t.\ $\C\in\WIN(S,\psi,\Map{})$ 
iff $\C$ is feedback-composable with $S$ and $\Map{\CPaths{S, \C}}\subseteq\semantics{\psi}$. %We note that the later subset-formalization implicitly requires that the closed-loop system never blocks. %That is, for every path $\pi\in\CPaths{S, \C}$, we need that $\C(\Ext{\pi_{[0;k]}})\neq\emptyset$ for all $k\in\N$.

% $\Map{\CPaths{S, \C}}\subseteq\semantics{\psi}$.
% We define 
% $\WIN(S,\psi,\Map{}):= \set{\C\mid \Map{\CPaths{S, \C}}\subseteq\semantics{\psi}}$ as the set of all such output-feedback control strategies.
% For a \emph{state-feedback} controller $\C^\dagger$, we define analogously the set 
% $\WIN^\dagger(S,\psi,P)$.

% 
% (a) state-feedback control with state-based specification (S-S), (b) state-feedback control with input/output-based specification (S-I/O), (c) output-feedback control with state-based specification (O-S), and (d) output-feedback control with input/output-based specification (O-I/O). The respective control problems ask to find a controller $\C$ s.t.\ (a)  $\CPaths{S, \C^\dagger}\subseteq\semantics{\psi^\dagger}$, (b) $\Ext{S, \C^\dagger}\subseteq\semantics{\psi}$, (c) $\CPaths{S, \C}\subseteq\semantics{\psi^\dagger}$, and (d) $\Ext{S, \C}\subseteq\semantics{\psi}$. We denote the set of all such sound control strategies as (a) $\WIN^\dagger(S,\psi^\dagger)$, (b) $\WIN^\dagger(S,\psi)$, (c) $\WIN(S,\psi^\dagger)$, and (d) $\WIN(S,\psi)$.
% %

\begin{figure}[t]
	\centering
	\includegraphics[width =  0.95 \columnwidth]{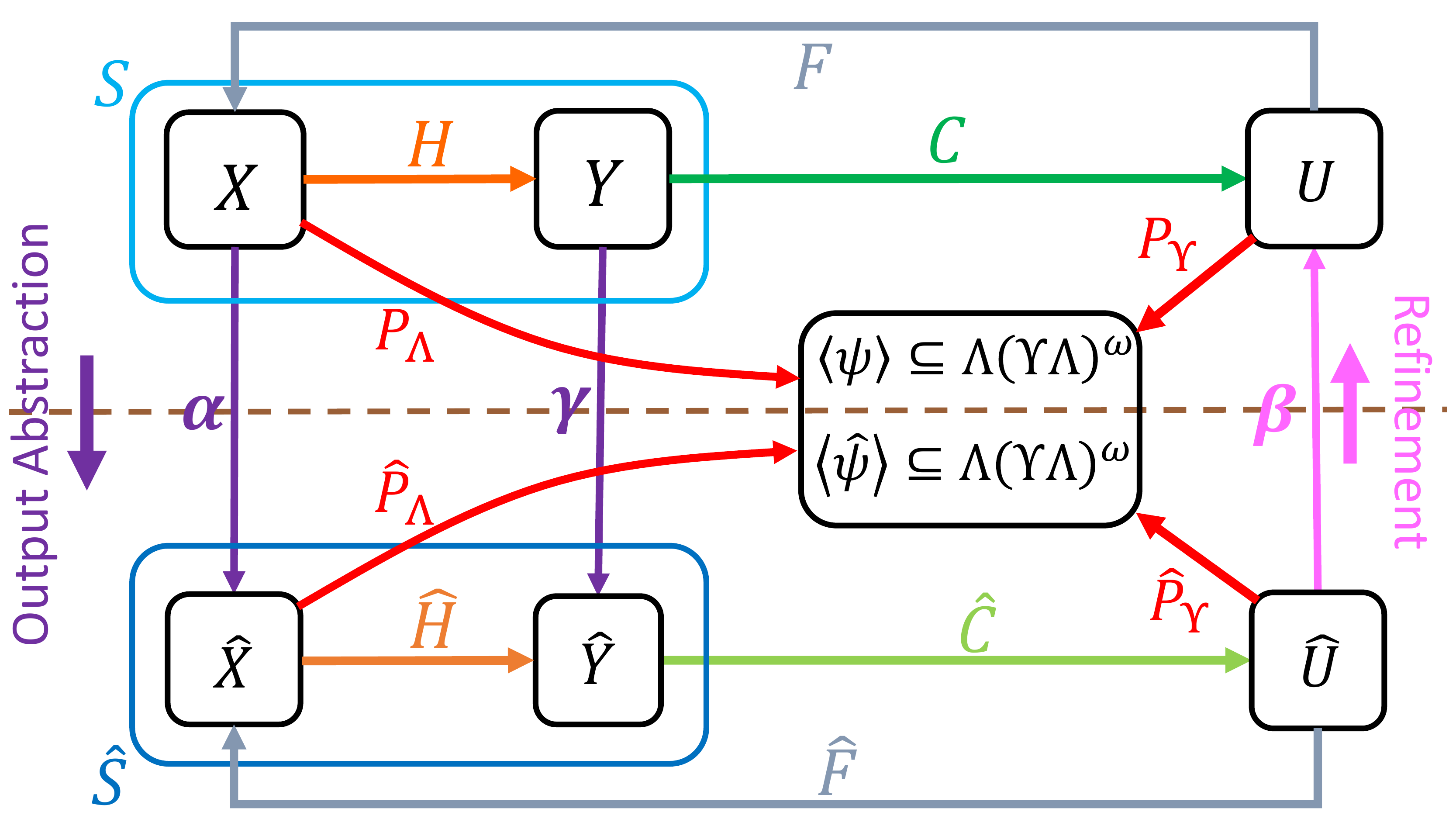}
\vspace{-0.3cm}
	\caption{Schematic representation of our ABoCD framework where $\Q=\tuplel{\alpha,\beta,\gamma}$ needs to be an eFRR (Def.~\ref{def:SoundAbs}) and $(\pa,\widehat{P})$ needs to be a sound abstract specification (Def.~\ref{def:soundaspec}). }
	\vspace{-0.3cm}
	\label{fig:sec2_outline}
\end{figure}

\section{Sound Abstraction-Based Output-Feedback Control}\label{sec:soundabs}
Within this section we extend the notion of feedback-refinement relations (FRR) and sound abstract specifications from \cite{ReissigWeberRungger_2017_FRR} to systems with possibly infinite state, input and output spaces and non-deterministic observation and predicate maps.
The intuition behind this extension is depicted in \REFfig{fig:sec2_outline}. Instead of a single relation between $X$ and $\Xa$ as in FRR, our \emph{extended feedback refinement relation (eFRR)} is a tuple $\Q=\tuplel{\alpha,\beta,\gamma}$ of relations between the tuples $(X,U,Y)$ and $(\Xa,\Ua,\Ya)$ (see the purple and pink arrows in \REFfig{fig:sec2_outline}). In addition, the notion of \emph{sound abstract specifications} ensures that the specification is correctly abstracted and interpreted over $\Sa$ via $\Mapa{}$ (indicated in red in \REFfig{fig:sec2_outline}) for systems related via $\Q$. 
As the main result of this section, we show that this extended notion of FRR together with sound abstract specifications allows for sound abstraction-based output-feedback control design (ABoCD) in the presents of non-deterministic observation and predicate maps.

\smallskip
\noindent\textbf{Sound Abstractions.}
In direct analogy to the definition of feedback-refinement relations (FRR) in \cite[Def. V.2]{ReissigWeberRungger_2017_FRR} we define a \emph{sound abstraction} under an \emph{extended feedback refinement relation (eFRR)} as follows.  

\begin{definition}\label{def:SoundAbs}
Let $S = (X, \Xo, U, F, Y, H)$ and $\Sa = (\Xa, \Xoa, \Ua, \Fa, \Ya, \Ha)$ be systems. % with observation maps $\Omap{}:Y\fun P$ and $\Omapa{}:\Ya\fun P$, respectively.
Further, let $\Q=\tuplel{\alpha,\beta,\gamma}$ be a set of set-valued functions s.t.\
$\alpha:X\setfun \Xa$, $\beta:\Ua\setfun U$ and $\gamma:Y\setfun \Ya$. 
Then we call $\Sa$ a \emph{sound abstraction}  of $S$ under $\Q$, written 
$S\frr{\Q}{} \Sa$, if
\begin{compactenum}
\item[$\ON{(A1)}$] $\AllQ{x\in \Xo}{\emptyset\neq\alpha(x) \subseteq \Xoa}$, 
\end{compactenum}
and for all $x\in X$ and $\xa\in\alpha(x)$ holds that 
\begin{compactenum}
\item[$\ON{(A2)}$] for all $\ua\in \Enab_{\Sa}(\xa)$ holds 
\begin{compactenum}[(i)]
\item $\emptyset\neq\beta(\ua)\subseteq \Enab_{S}(x)$ and
\item $\emptyset\neq\alpha(F(x,\beta(\ua)))\subseteq \Fa(\xa,\ua)$, and
\end{compactenum}
\item[$\ON{(A3)}$] $\emptyset\neq\gamma(H(x))\subseteq \Ha(\xa)$.
\end{compactenum}
$\Sa$ is  a \emph{sound realization} of $S$, written $S\frrE{\Q}{} \Sa$, if $S\frr{\Q}{} \Sa$ and $\Sa\frr{\Q^{-1}}{} S$ where $\Q^{-1}:=\tuplel{\alpha^{-1},\beta^{-1},\gamma^{-1}}$.
\end{definition}

In analogy to \cite[Def. V.2]{ReissigWeberRungger_2017_FRR} we call $\Q$ an \emph{extended feedback refinement relation (eFRR)} from $S$ to $\Sa$. We write $S\frr{}{} \Sa$ if there \emph{exists} an eFRR $\Q$ s.t.\ $S\frr{\Q}{} \Sa$.

Utilizing \REFfig{fig:sec2_outline} we can interpret \REFdef{def:SoundAbs} as follows. For $\Q=\tuplel{\alpha,\beta,\gamma}$ to be an eFRR we require that
for every state $x$ of $S$ and every abstraction $\xa\in\alpha(x)$ holds that all refinements $u\in\beta(\ua)$ of the inputs $\ua$ enabled in $\xa$ are enabled in $x$ (
$\ON{(A2.i)}$). This prevents deadlocks in $S$, i.e., any input choice made by the abstraction and any non-deterministic refinement choice of this input allows progress in $S$. Further, under a transition of such related enabled inputs in related states the original system $S$ can only reach states $x'$ where all related abstract states $\xa'$ are also reachable in $\Sa$ from $\xa$ under $\ua$ ($\ON{(A2.ii)}$). This ensures that transitions of the abstract system always overapproximate any related behavior of the original system. 
Similarly, $\ON{(A3)}$ ensures that all possible observations $y$ of a state $x$ are only related to abstract observations $\ya$ that are observable in the related state $\xa$. This ensures that the observed behavior of $\Sa$ overapproximates the true observations for output-feedback control.
Finally, condition $\ON{(A1)}$ together with $\ON{(A2.ii)}$ ensures that all reachable states of $S$ are related to at least one abstract state. 

\smallskip
\noindent\textbf{Sound Abstract Specifications.}
In order to ensure that the specification is correctly interpreted over $\Sa$ via $\Mapa{}$ we next introduce the concept of sound abstract specifications in analogy to \cite[Def. VI.2]{ReissigWeberRungger_2017_FRR}.

\begin{definition}\label{def:soundaspec}
Let $S$ and $\Sa$ be systems s.t.\ $S\frr{\Q}{} \Sa$. 
Further, let $(\psi,\Map{})$ and $(\pa,\Mapa{})$ be specifications interpreted over $S$ and $\Sa$ respectively. Then we say that $(\pa,\Mapa{})$ is a \emph{sound abstract specification} associated with $S$, $\Sa$, $\Q$ and $(\p,\Map{})$, written $\tuplel{S,(\psi,\Map{})}\frr{\Q}{} \tuplel{\Sa,(\pa,\Mapa{})}$ if the following holds. 
For \emph{all} $\pi=x_0u_0x_1\hdots\in\CPaths{S}$ for which there \emph{exists} an input sequence $\ua_0\ua_1\hdots$ with $u_i\in\beta(\ua_i)$ (for all $i\in\N$) s.t.\ for \emph{all} $\pia=\xa_0\ua_0\xa_1\hdots$ with $\xa_i\in\alpha(x_i)$ (for all $i\in\N$) holds $\Mapa{\pia}\subseteq\semantics{\pa}$ also holds that $\Map{\pi}\subseteq\semantics{\psi}$.
\end{definition}

Intuitively, \REFdef{def:soundaspec} ensures that for every \enquote{good} path $\widehat{\pi}$ over $\Sa$, i.e., a path that only generates predicate sequences in $\semantics{\pa}$, all paths of $S$ related to $\widehat{\pi}$ via the eFRR $Q$ are also \enquote{good}, i.e., only generate predicate sequences in $\semantics{\psi}$. If this holds, we can use $(\pa,\Mapa{})$ to synthesize a \enquote{good} abstract controller which can then be refined into a \enquote{good}  controller $\C$ for $\Sa$. This is formalized next.

\smallskip
\noindent\textbf{Sound Controller Refinement.}
As the main result of this section, we now show how output-feedback controllers for sound abstractions under sound abstract specifications can be refined to output-feedback controllers for the original system w.r.t.\ the original specification.

\begin{theorem}[Sound ABoCD]\label{thm:ABoCD}
 Let $\tuplel{S,(\psi,\Map{})}\frr{\Q}{} \tuplel{\Sa,(\pa,\Mapa{})}$ with $\Q=\tuplel{\alpha,\beta,\gamma}$.
 Further, let $\Ca\in\WIN(\Sa,\pa,\Mapa{})$ and define $\C$ s.t.\ 
 \begin{align}\label{equ:ABoCD}
  \AllQ{\sigma\in\EPrefs{S}}{\C(\sigma)\in\beta(\Ca(\Omega_{\beta,\gamma}(\sigma)))},
 \end{align}
 where $\ya_0\ua_0\ya_1\hdots\in \Omega_{\beta,\gamma}(y_0u_0y_1\hdots)$ iff $\ya_k\in\gamma(y_k)$ and $u_k\in\beta(\ua_k)$ for all $k\in\N$. Then $\C\in\WIN(S,\psi,\Map{})$.
\end{theorem}

In order to prove \REFthm{thm:ABoCD} we need to show that the constructed controller $\C$ is feedback-composable with $S$ and only generates paths that are compatiple with the specification $\psi$. Intuitively, this requires to show that at every time step $k$ every input choice $u_k$ made by $\C$ via \eqref{equ:ABoCD} based on the already observed external sequence $\sigma_{k}=y_0u_0\hdots y_{k}$ ensures that $u_k$ is enabled in all possible states reached under this observation and that all possible paths of $S$ that conform with these observation trances are compliant with $\semantics{\psi}$ via $\Map{}$. In order to prove this claim, we first formalize some observations about all possibly generated $\sigma_{k}$ and all compliant paths $\pi_{k}\in\Extn{S}^{-1}(\sigma_{k})$ that result from the fact that $S\frr{\Q}{}\Sa$ in the following lemma, which is proven in \REFapp{sec:app:ABoCD}. 

\begin{lemma}\label{lem:ABoCD}
  Given the premises of \REFthm{thm:ABoCD} the following holds for all $k\in\N$. For all $\sigma_{k-1}=y_0u_0\hdots y_{k-1}\in\EPrefsk{S,\C}{k-1}$, $u_{k-1}\in\C(\sigma_{k-1})$ (if $k>0$) and $\sigma_{k}=y_0u_0\hdots y_{k-1}u_{k-1}y_{k}\in\EPrefs{S}$ holds that 
  \begin{compactenum}[(a)]
   \item $\sigma_k\in\EPrefsk{S,\C}{k}$,
  \item for all $\sigmaa_{k}=\ya_0\ua_0\hdots\ya_{k}\in\Omega_{\beta,\gamma}(\sigma_{k})$ s.t.\ $\ua_{i}\in\beta^{-1}(u_{i})\cap\Ca(\sigmaa_{i})\neq\emptyset$ for all $i\in[0;k-1]$ holds  that $\sigmaa_{k}\in\EPrefsk{\Sa,\Ca}{k}$,
   \item for all $\pi_{k}=x_0u_0\hdots x_{k-1}u_{k-1}x_k\in\Extn{S}^{-1}(\sigma_{k})$, $\sigmaa_k$ as in (c), and $\pia_k=\xa_0\ua_0\hdots\ua_{k-1}\xa_k\in\Omega_{\alpha,\beta}(\pi_k)$  s.t.\ the input sequence $\ua_0\hdots\ua_{k-1}$ matches the inputs of $\sigmaa_{k}$, holds that $\pia_k\in\Extn{\Sa}^{-1}(\sigmaa_{k})$, 
   \item $\emptyset\neq\C(\sigma_k)\in\Enab_{S}(\LastSn{S}{\sigma_k})$.
  \end{compactenum}
\end{lemma}

Intuitively, \REFlem{lem:ABoCD} shows that no matter how the non-determinism in the formulation of \eqref{equ:ABoCD} is resolved, the resulting control input to $S$ ensures that this system only stays on paths that are related to paths of $(\Sa,\Ca)$.
I.e., given an observation $\sigma_k$ then any projection of this sting to $\Sa$ via $\Omega_{\beta,\gamma}$ that results in a string that is actually possible in $(\Sa,\Ca)$ (and we know that at least one such sting exits, i.e., the one that corresponds to the observation made in $\Sa$ while generating the corresponding inputs for instances $i<k$) all abstract inputs enabled by $\Ca$ and all possible refinements of this input via $\beta$ are actually enabled in $S$. I.e., letting $\C$ choose such an input results in a non-blocking behaviour of the closed loop. Then the definition of eFRR ensures that the resulting traces always stay related. With this intuition it is not surprising that \REFlem{lem:ABoCD} allows us to prove \REFthm{thm:ABoCD} under the assumption that $(\pa,\Mapa{})$ is a sound abstract specification as in \REFdef{def:soundaspec}.

\begin{proof}
It immediately follows from \REFlem{lem:ABoCD} (a) and (d) that $\C$ is feedback-composeable with $S$ and it remains to show that $\Map{\CPaths{S, \C}}\subseteq\semantics{\psi}$. 
We recall from the definition of $\CPaths{S,\C}$ that $\pi=x_0u_0x_1\hdots\in\CPaths{S,\C}$ iff $|\pi|=\infty$ and $\pi|_{[0;k]}\in\Extn{S}^{-1}(\EPrefsk{S,\C}{k})$ for all $k\in\N$. Now it follows from \REFlem{lem:ABoCD} (b/c) that there exists an input sequence $\ua_0\ua_1\hdots$ s.t.\ for all $\pia=\xa_0\ua_0\xa_1\hdots$ with $\xa_k\in\alpha(x_k)$ holds that $\pia|_{[0;k]}\in\Extn{\Sa}^{-1}(\EPrefsk{\Sa,\Ca}{k})$ and therefore $\pia\in\CPaths{\Sa,\Ca}$. As $\Ca\in\WIN(\Sa,\pa,\Mapa{})$ we further have 
$\Mapa{\CPaths{\Sa, \Ca}}\subseteq\semantics{\pa}$ and therefore $\Mapa{\pia_k}\subseteq\semantics{\pa}$. With this, it follows from \REFdef{def:soundaspec} that $\Map{\pi}\subseteq\semantics{\psi}$, what proves the claim.
\end{proof}

\smallskip
\noindent\textbf{Algorithmic ABoCD.}
So far, we have  defined a sound ABoCD framework for an output feedback control problem $\tuplel{S,\psi,\Map{}}$. In the remainder of this paper we will target the problem of algorithmically computing
\begin{compactenum}[(I)]
 \item a \emph{sound abstraction} $\Sa$ and a \emph{sound abstract specification} $(\pa,\Mapa{})$ s.t.\ 
% \begin{equation}
 $\tuplel{S,(\psi,\Map{})}\frr{}{} \tuplel{\Sa,(\pa,\Mapa{})}$,
% \end{equation}
 \item  an output feedback controller $\Ca\in\WIN(\Sa,\pa,\Mapa{})$.
\end{compactenum}
 If we solve these two algorithmic challenges, we can apply \REFthm{thm:ABoCD} to obtain a sound controller  $\C\in\WIN(S,\psi,\Map{})$ for the original ABoCD problem via \eqref{equ:ABoCD}. 
 
 We first discuss step (I) in \REFsec{sec:Abst:A} for a particular class of systems $S$, which is a straightforward extension of the constructions in \cite{ReissigWeberRungger_2017_FRR} which are implemented in the tool \texttt{SCOTS} \cite{Scots}. Then, in \REFsec{sec:Abst:B}, we provide the second main contribution of this paper, which is solving step (II) by utilizing \emph{bounded synthesis} \cite{schewe2007bounded} which is implemented in the tool \texttt{BoSy} \cite{faymonville2017bosy}.

\section{Constructing Sound Finite Abstractions}\label{sec:Abst:A}

Within this section we follow the grid-based abstraction technique developed in \cite{ReissigWeberRungger_2017_FRR} for non-linear systems with disturbances. 
This abstraction process starts with a continuous non-linear \emph{control system} $\Sigma$ which is first time-discretized into a \emph{system} $S$ of the form introduced in \REFsec{sec:prelim}. This system $S$ has infinite input, state and output spaces. It is therefore  further abstracted into a \emph{finite system} $\Sa$ which can be used for symbolic controller synthesis.

In the following we recall this abstraction process from \cite{ReissigWeberRungger_2017_FRR} an discuss the special treatment of actuation and observation constrains and the abstraction of the specification. We show that our definition of sound abstractions (Def.~\ref{def:SoundAbs}) and sound abstract specifications (Def.~\ref{def:soundaspec}) is readily fulfilled by this abstraction procedure.

\smallskip
\noindent\textbf{Control System.}
A \emph{control system} $\Sigma = (X, \Xo, U, W, f,Y,h)$ consists of a continuous state space $X= \real{n}$, a set of initial states $\Xo\subseteq X$, a non-empty compact set of inputs $U\subseteq\real{m}$, a continuous output space $Y\subseteq\real{r}$, a compact cell $W\subseteq X$, and nonlinear (differential) inclusions% of the form
\begin{subequations}\label{equ:def_fh}
 \begin{align}
 \dot{\xi}&\in f(\xi(t),\mu(t))+W~\text{and}\label{equ:def_f}\\
  \nu(t)&\in h(\xi(t)),\label{equ:def_h}
\end{align}
\end{subequations}
where $f(\cdot,u)$ is locally Lipschitz for all $u\in U$.

\smallskip
\noindent\textbf{Continuous Transition System.}
A control system can be \emph{time-discretized} to obtain a system $S$ as defined in \REFsec{sec:prelim}.  
I.e., given a time sampling parameter $\tau>0$, we can define the system $S=(X,X,U,F,Y,H)$ associated with $\Sigma$ as follows.
First, given an initial state $\xi(0)\in X$, and a constant input trajectory $\mu_u:[0,\tau]\rightarrow U$ which maps every $t\in [0,\tau]$ to the same $u\in U$, 
a solution of the inclusion in \eqref{equ:def_f} 
on $[0,\tau]$ is an absolutely continuous function $\xi:[0,\tau]\rightarrow X$  
that fulfills \eqref{equ:def_f} for almost every $t\in[0,\tau]$. We collect all such solutions in the set $\ON{Sol}_f(\xi(0),\tau,u)$. 
Then the transition and output functions of $S$ are defined s.t.\ for 
all $x\in X$ and for all $u \in U$ it holds that $x'\in F(x,u)$ and $y'\in H(x')$ iff there exists a solution $\xi\in\ON{Sol}_f(x,\tau,u)$ 
s.t.\ $\xi(\tau)=x'$ and $y'\in h(x')$. 

\smallskip
\noindent\textbf{Finite Abstract System.}
Following \cite{ReissigWeberRungger_2017_FRR} one can now apply a grid-based discretization of the state space of $S$ to obtain a system $\Sa$ with finitely many states. For this, one usually fixes a grid parameter $\eta\in\real{}_{>0}^n$ and a region of interest defined as a hyper-rectangle $\Theta = \hyint{\alpha, \beta}$, s.t.\ $\beta - \alpha$ is an integer multiple of $\eta$. Then one defines the finite abstract state space $\Xa$ as a set of hyper cells which cover $\Theta$ with grid-aligned cells $\hyint{a,b}$ s.t.\ $b-a=\eta$, while covering the rest of the state space with \enquote{overflow-cells} of the form $\hyint{
\{-\infty\}^n,\alpha}$ and $\hyint{\beta,\{\infty\}^n}$. These \enquote{overflow-cells} are then treated as obstacles and added to the specification.

For the discretization of the output space $Y$ one can impose a very similar grid-based discretization with a possibly different grid parameter $\eta'$. This would allow us to capture the example in \REFfig{fig:EX2}. On the other hand, we can also consider the case where given observation constrains impose a finite set of observations $\Ya$ (as in \REFfig{fig:EX1}). In both cases, $\Ya$ is a finite cover of $Y$.

To discretize the input space, one usually restricts attention to a finite subset of \enquote{representative} inputs $\Ua\subseteq U$. This conforms, on one hand, with a grid-based discretization of inputs (similar to $X$ and $Y$) and picking one representative per grid cell\footnote{Our framework also allows to capture imprecise actuation, (i.e., an actuator which always has an error bounded by some $\varepsilon$). This would result in a non-trivial refinement map $\beta$ in \REFthm{thm:abstraction} and would require some adjustments in the definition of $\Sa$ that we omit due to space constrains.}.

Given these finite state, input and output sets one can define the \emph{finite abstract system} $\Sa=(\Xa,\Xa,\Ua,\Fa{},\Ya,\Ha{})$ of $S$ s.t.\footnote{We use the technique explained in \cite{ReissigWeberRungger_2017_FRR} and implemented in 
\texttt{SCOTS} \cite{Scots} to over-approximate the set $\set{\cup_{x\in\xa}\ON{Sol}_f(x,\tau,\ua)}$.}
$\xa'\in\Fa{}(\xa,\ua)$ iff $\set{\cup_{x\in\xa}\ON{Sol}_f(x,\tau,\ua)} \cap \xa' \neq \emptyset$, and $\ya\in\Ha(\xa)$ iff $H(\xa)\cap\ya\neq\emptyset$. 

In addition to the abstract system, we define the maps $\Imapa{}$ and $\Omapa{}$ to interpret $\psi$ over $\Sa$ s.t.\
$\mu\in\Imapa{\ua}$ iff $\mu\in\Imap{\ua}$ and $\lambda\in\Omapc{\xa}$ iff $\lambda\in\Omap{\xa}$. In this case we have $\pa:=\psi$.
% % 

\smallskip
\noindent\textbf{Soundness.}
We have the following expected result on the soundness of the outlined abstraction procedure. % which is proven in \cite{SchmuckZareian_CDC21_extended}.

\begin{theorem}\label{thm:abstraction}
 Let $S$ be the time-discretized system associated with the control system $\Sigma$ and $\Sa$ its grid-based abstraction. Further, let $\psi$ be an LTL specification interpreted over $S$ and $\Sa$ via the maps $\Imap{}$ and $\Omap{}$, and $\Imapa{}$ and $\Omapa{}$, respectively. Then
 $$\tuplel{S,(\psi,\Map{})}\frr{}{} \tuplel{\Sa,(\psi,\Mapa{})}. $$
\end{theorem}

\begin{proof}
 We define $\alpha$, $\beta$ and $\gamma$ s.t.\ $\xa\in\alpha(x)$ iff $x\in\xa$, $\beta(\ua)=\ua$ and $\ya\in\gamma(y)$ iff $y\in\ya$. We first prove that $S\frr{Q}{}\Sa$ by showing that (A1)-(A3) in \REFdef{def:SoundAbs} holds.\\
 \begin{inparaitem}[$\blacktriangleright$]
  \item (A1) Follows from the fact that the initial states are not restricted (i.e., $\Xo=X$ and $\Xao=\Xa$) and $\Xa$ is a cover of $X$.\\
  \item (A2.i) Follows from the fact that $S$ is input enabled by definition, i.e., $\Enab{S}(x)=U$ for all $x\in X$. With this, it follows from the definition of $\Fa$ that $\Sa$ is also fully input enabled, i.e., $\Enab{\Sa}(\xa)=\Ua$ for all $\xa\in\Xa$. With this, the claim directly follows.\\
  \item (A2.ii) We pick $\ua\in\Ua$ and observe that $\beta(\ua)=\ua$ by construction. Now it follows from the construction of $\Fa$ that for every $\xa\in\alpha(x)$ we have $\xa'\in\Fa(\xa,\ua)$ if there exists some $x'\in F(x,\ua)$ s.t.\ $x'\in\xa'$. This immediately implies the claim.\\ 
  \item (A3) Follows directly from the definition of $\Ha$.\\
 \end{inparaitem}
It remains to show that $(\psi,\Mapa{})$ is a sound abstract specification. As $\pa=\psi$ and $\beta(\ua)=\ua$ this, however, follows immediately from the definition $\Mapa{}$.
\end{proof}

\section{Output-Feedback Controller Synthesis for State-Based Specifications }\label{sec:Abst:B}

Within this section we consider a system $\Sa=(\Xa,\Xao,\Ua,\Fa,\Ya,\Ha)$ with \emph{finite} sets $\Xa$, $\Ua$ and $\Ya$ which allows to interpret the \emph{original LTL specification} $\psi$ over $\Sa$ via maps $\Imapa{}:\Ua\setfun\I$ and $\Omapa{}:\Xa\setfun\O$, defining the path map $\Mapa{}$. 
Our  goal is to construct an \emph{output-feedback controller} $\Ca\in\WIN(\Sa,\psi,\Mapa{})$. 

In order to formalize this construction we need to introduce some additional notation.

\subsection{Preliminaries}

\smallskip 
\noindent\textbf{Transition Systems.}
We define finite transition systems as tuples $T=\tuple{Q,Q_0,\Sigma,\delta}$ where $Q$ is a finite set of states, $Q_0\subseteq Q$ is a set of initial states, $\Sigma$ is a finite alphabet and $\delta: Q\times\Sigma\setfun Q$ is a set valued transition function. We call $T$ deterministic if $|\delta(q,\sigma)|\leq 1$ for all $q\in Q$ and $\sigma\in\Sigma$.

% and $\acc\subseteq Q$ is a set of rejecting states. 
A sequence $\pi=q_0\sigma_0q_1\sigma_1\hdots$ is a (maximal) path over $T$ if $q_0\in Q_0$ and $q_{k+1}\in\delta(q_k,\sigma_k)$ for all $k\in \dom{\pi}$ and either $|\pi|=\infty$ or $\delta(q_k,\sigma)=\emptyset$ for all $\sigma\in\Sigma$. Similar to systems we collect all maximal paths of $T$ in the set $\Paths{T}$ and their corresponding external sequences $\sigma_0\sigma_1\hdots$ in the set $\EPaths{T}$. 

Given an infinite string $\sigma=\sigma_0\sigma_1\hdots\in\Sigma^\omega$ we say that a path $\pi$ over $T$ is \emph{compliant} with $\sigma$ if $q_{k+1}\in\delta(q_k,\sigma_k)$ for all $k\in\dom{\pi}$ and either (i) $|\pi|=\infty$ or (ii) $|\pi|=k<\infty$ and $\delta(q_k,\sigma_k)=\emptyset$. 
We say that $T$ is \emph{complete} if for all $\sigma\in\Sigma^\omega$ exists at least one \emph{infinite} path over $T$ compliant with $\sigma$.
 
\smallskip 
\noindent\textbf{Universal Co-Büchi Automata.}
We define automata as tuples $\Aut=\tuple{T,\acc}$ where $T$ is a finite transition system and $\acc\subseteq Q$ is a set of rejecting states. %We transfer the definitions of paths and external sequences from $T$ to $\Aut$. 
A universal co-Büchi automata (UCA) is an automaton where $\mathcal{F}$ is interpreted as a \emph{Co-Büchi condition} and external sequences $\sigma$ are accepted based on \emph{universal branching}. 
I.e., a path $\pi$ of $T$ is accepted by a UCA $\Aut$ if $\pi$ visits the set $\mathcal{F}$ only \emph{finitely} often. Further, an external sequence $\sigma$ of $T$ is accepted by a UCA $\Aut$ iff \emph{all} runs $\pi$ compliant with $\sigma$ are accepted by $\Aut$\footnote{This is in contrast to non-deterministic acceptance, where only one compliant run needs to be accepting.}. 
We collect all \emph{infinite} external sequences that are accepted by a UCA $\Aut$ in its language $\Lang(\Aut)$. 

\smallskip 
\noindent\textbf{Models of UCAs.}
Let $M$ be a \emph{deterministic} transition system. Then we say that $M$ is a \emph{maximal model} of a UCA $\Aut$ whenever $\pi\in\EPaths{M}$ implies
$\pi\in\EPaths{T}$ s.t.\ $\pi$ is compliant with some $\sigma\in\Lang(\Aut)$ and there exists no $\pi'\in\EPaths{T}$ s.t.\ $\pi$ is a prefix of $\pi'$.

\begin{figure}
	\centering
	\includegraphics[width = 0.7 \columnwidth]{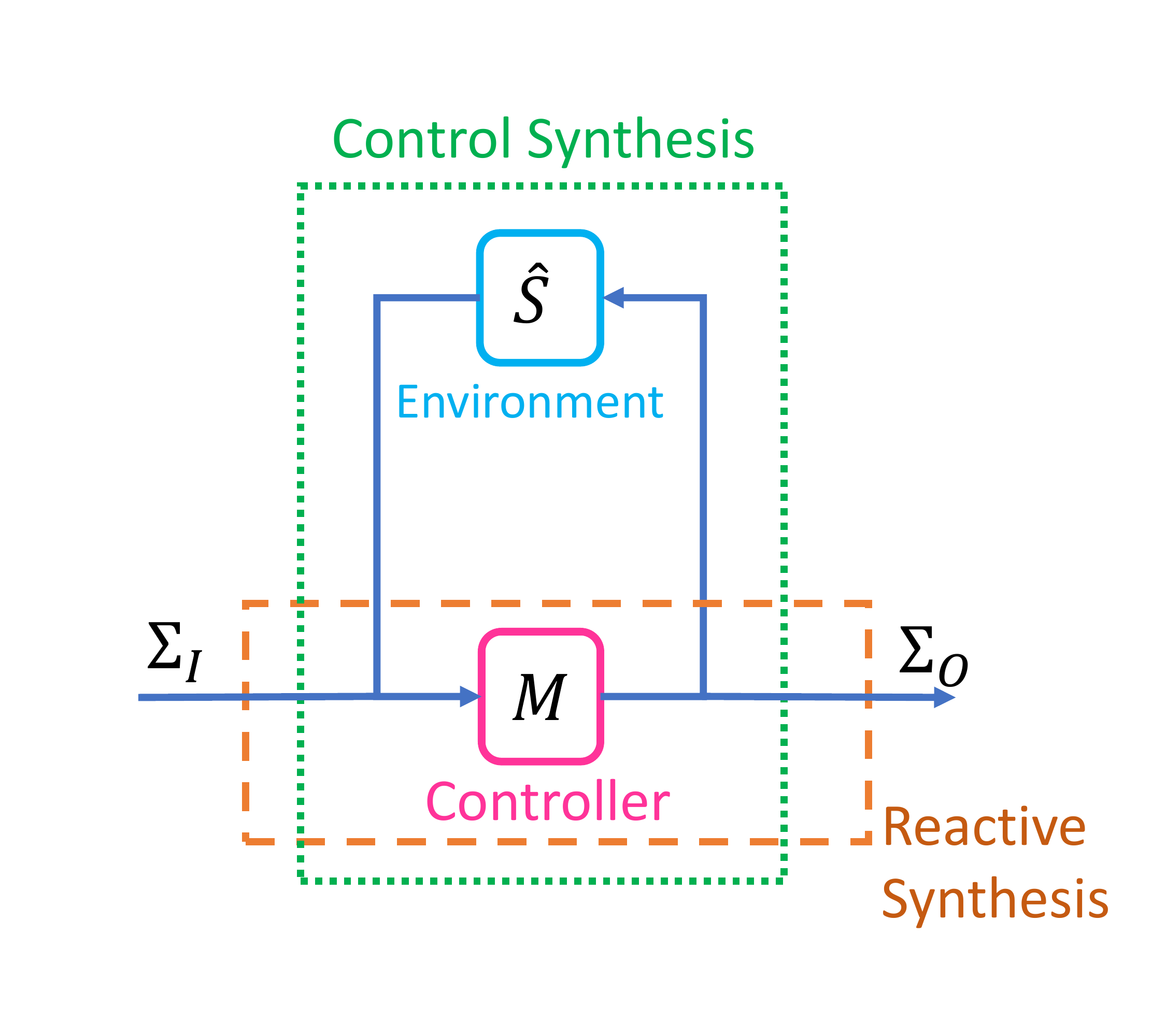}
	\vspace{-0.6cm}
	\caption{Schematic representation of the correspondence between (a) the solution of an LTL realizability problem via reactive synthesis (dashed brown) resulting in a deterministic transition system $M$ which only generates sequences in $\semantics{\psi}$ while receiving arbitrary inputs from an unrestricted environment, and
	(b) the controller synthesis problem $\tuplel{\Sa, \psi,\widetilde{P}}$  (dotted green) where the environment of $M$ is restricted by the dynamics of $\Sa$. The challenge in combining these two problems lies in the fact that for $\Sa$ we have $\Sigma_I=\Ua$ and $\Sigma_O=\Ya$ and for $M$ we have $\Sigma_I=\O$ and $\Sigma_O=\I$.  }
	\vspace{-0.6cm}
	\label{fig:Syn_outline}
\end{figure}

\subsection{The Realizability Problem for LTL formulas.}
In the field of reactive synthesis, a standard, well understood problem is to compute a so called \emph{reactive module} that \emph{realizes} a specification $\psi$ given in LTL against any input sequence imposed by some (unknown) environment (see the dashed box in \REFfig{fig:Syn_outline}). 
More formally, a solution to this realizability problem is a deterministic transition system $M$ which only generates sequences in $\semantics{\psi}$ while receiving arbitrary inputs from an unrestricted environment, i.e., assuming $\EPaths{\Sa}=\Sigma_I(\Sigma_O\Sigma_I)^\omega$ in the green dotted box in \REFfig{fig:Syn_outline}.

Such a solution to a given realizability problem can be computed in many different ways (see e.g. \cite[Ch.1]{Eherls_phdthesis} for an overview).
Within this paper, we use a rather modern approach, called \emph{Bounded Synthesis} \cite{schewe2007bounded}. This synthesis procedure avoids the computationally expensive step of determinizing a non-deterministic Büchi automaton via Safra's construction \cite{Safra_88}, and utilizes a direct synthesis technique via UCA, introduced in \cite{KupfermanVardi_safraless_2005}, instead. As a byproduct, bounded synthesis generates small transition systems $M$ (in terms of their state size) if they exist. We are summarizing the necessary features and constructions of \emph{Bounded Synthesis} in the next proposition. 

\begin{proposition}{\cite[Thm.5.1]{KupfermanVardi_safraless_2005}}\label{prop:Auts}
 Let $\psi$ be an LTL formula over a finite set of atomic propositions $\AP=\AP_I\dot{\cup}\AP_O$ and define $\I:=2^{\AP_I}$ and $\O:=2^{\AP_O}$. Then one can construct a \emph{complete} UCA $\Auts=\tuple{Q,Q_0,\I\times\O,\delta,\mathcal{F}}$ with $2^{O(|\psi|)}$ number of states and with language $\Lang(\Auts)=\semantics{\psi}$. We call $\Auts$ the UCA induced by $\psi$.
\end{proposition}

\begin{proposition}[Bounded Synthesis \cite{schewe2007bounded}]\label{prop:Realize}
Let $\Aut$ be a UCA. Then there exists an algorithm to compute a maximal model $M$ of $\Aut$ if it exists. 
\end{proposition}

\subsection{Output-Feedback Controller Synthesis}
The main obstacle in utilizing reactive synthesis to solve the \emph{output-feedback control problem} $\tuplel{\Sa, \psi,\Mapa{}}$, as indicated by the dotted green box in \REFfig{fig:Syn_outline}, lies in the fact that for $\Sa$ we have $\Sigma_I=\Ua$ and $\Sigma_O=\Ya$ and for $M$ we have $\Sigma_I=\O$ and $\Sigma_O=\I$ (as $\semantics{\psi}\subseteq\O(\I\O)^\omega)$ while their relationship is in general \emph{non-deterministic}. 

In the usual setting of ABCD with state-feedback and deterministic predicate maps (as in \cite{ReissigWeberRungger_2017_FRR}) or ABoCD with output-feedback and specifications defined over the observables (as, e.g.\ in \cite{MajumdarNS2020,khaled2020outputfeedback}), the relationship of the inputs and outputs of $\Sa$ and $M$ is \emph{deterministic}. In this case, utilizing reactive synthesis for control design is computationally simpler and done as follows. 

First, one constructs a deterministic observer of $\Sa$ over $\Ua$ and $\Ya$ (which is at most exponential in the state-size of $\Sa$), by using for example subset construction. Second, one uses the usual machinery in reactive synthesis to translate the specification $\psi$ first into a non-deterministic Büchi automaton (NBA) which is then determinized into a deterministic two-player game (again at most exponential in the state-size of the NBA). 
In the third and last step, the deterministic observation system determined from $\Sa$ is combined with this game. Due to the \emph{determnistic} relationship of all involved variables the resulting game is still deterministic and can be solved with common techniques from reactive synthesis.

In the presence of \emph{non-determnistic} relationships between the alphabets of the observer automaton and the game, the combination of both yields a \emph{non-deterministic} game that must again be determinized before it can be solved. This causes another exponential blow-up (see e.g.\ \cite{chatterjee2010complexity} for a discussion of this aspect).

Instead of taking this three-step approach, we first combine $\Sa$ and the UCA $\Auts$ induced by $\psi$ into a new UCA $\Autap$ over their product state space and then use bounded synthesis to extract a model for this product UCA. While this requires a determinization step which is exponential in the size $\Autap$, only one such step is required.
The main contribution of this section is to show that this procedure leads to a model $M$ over $\Ya\times\Ua$ that allows to extract a sound controller $\Ca\in\tuplel{\Sa, \psi,\Mapa{}}$. 

\smallskip
\noindent\textbf{Combining $\Sa$ and $\Auts$.} We first describe the formal combination of $\Sa$ and $\Auts$ and then discuss the properties of the resulting UCA.

\begin{definition}\label{def:synthUCA}
Let $\Sa=(\Xa,\Xao,\Ua,\Fa,\Ya,\Ha)$ be a finite system with predicate maps $\Imapa{}:\Ua\setfun \I$ and $\Omapa{}:\Xa\setfun \O$. Further, let $\psi$ be an LTL formula with induced complete UCA $\Auts=\tuple{Q,Q_0,\I\times\O,\delta,\mathcal{F}}$. Then we define the product of $\Sa$ and $\Auts$ as the  UCA $\Autap=\tuple{\Qa,\Qa_0,\Ya\times\Ua,\dap,\acca}$ s.t.\
\begin{compactitem}
 \item $\Qa=(\Xa\times Q)\cup\set{\bot}$
 \item $\Qa_0=\Xao\times Q_0$
\item $\acca=(\Xa\times \acc)\cup\set{\bot}$
\item $\{\bot\}=\dap(\qa,(\ya,\ua))$ iff either $\qa=\bot$ or $\qa=(\xa,q)$, $\ya\in\Ha(\xa)$ and $\Fa(\xa,\ua)=\emptyset$, and 
\item $\bot\neq(\xa',q')\in\dap((\xa,q),(\ya,\ua))$ iff $\ya\in\Ha(\xa)$, $\xa'\in\Fa(\xa,\ua)$ and there exists $\i\in\Imapa{\ua}$ and $\o\in\Omapa{\ya}$ s.t.\ $q'\in\delta(q,(\i,\o))$.
\end{compactitem}
We call $\Autap$ the UCA induced by $\Sa$, $\psi$ and $\Mapa{}$.
\end{definition}

Intuitively, the UCA $\Autap$ has the following properties.\\ 
\begin{inparaenum}[(i)]
 \item Every path $\pi$ of $\Sa$ corresponds to a set of paths in $\Autap$, each describing exactly one predicate sequence induced by $\pi$. This implies that an external sequence $\sigma\in\EPrefs{\Sa}$ is only accepted by $\Autap$ if all predicate sequences generated by any compliant path of $\sigma$ over $\Sa$ are allowed by the specification.\\ 
 \item Every blocking path of $\Sa$ is extended to an infinite path in $\Autap$ that visits $\bot$ infinitely often. This implies that only external sequences of $\Sa$ are accepted by $\Autap$ for which each chosen input $\ua_k$ is always enabled in all states reachable by any compliant path of $\sigma$ on $\Sa$.\\
\end{inparaenum}
In order to formalize this observation we define the set of blocking external sequences $\ON{IBlock}(\Sa)$ as follows. 

\begin{definition}\label{def:IBlock}
Given a system $\Sa$ we define its set of blocking sequences $\ON{IBlock}(\Sa)\subseteq \Ya(\Ua\Ya)^\omega$ as follows. An external sequence $\sigmaa=\ya_0\ua_0\ya_1\ua_1\hdots\in \Ya(\Ua\Ya)^\omega$ is said to be blocking on $\Sa$, i.e., $\sigmaa\in\ON{IBlock}(\Sa)$, iff there exists an index $k\in\mathbb{N}$ s.t.\ $\nu=\sigmaa|_{[0,k]}\in\EPrefs{\Sa}$ but $\ua_k\notin\Enab(\LastSn{\Sa}{\nu})$.
\end{definition}

We emphasize that the map $\Enab$ is lifted to sets of states using \emph{intersection}. Hence, a string $\sigma$ is called \emph{blocking on $\Sa$} as defined in \REFdef{def:IBlock}, if there exists some state
$q\in\LastSn{\Sa}{\nu}$ in which $u_k$ is not enabled, even if there might exist other states $q'\in\LastSn{\Sa}{\nu}$ which allow progress on $u_k$. 

We are now ready to formalize the above observations in the following proposition, which is proven in \REFapp{sec:app:synthUCA}.

\begin{proposition}\label{prop:synthUCA_prop}
Given the premises of \REFdef{def:synthUCA}, let $\sigmaa=\ya_0\ua_0\ya_1\ua_1\hdots\in \Ya(\Ua\Ya)^\omega$ be an external sequence. Then 
\begin{compactenum}[(i)]
\item $\sigmaa\in\EPaths{\Sa}\cap\Lang(\Autap)$ iff $\Mapa{\Extn{\Sa}^{-1}(\sigmaa)}\subseteq\semantics{\psi}$, %and 
\item $\sigmaa\in\Lang(\Autap)$ implies $\sigmaa\notin\ON{IBlock}(\Sa)$. 
\end{compactenum}
\end{proposition}

\smallskip
\noindent\textbf{Extracting a control strategy.}
As $\Autap$ is a UCA over $\Ya\times\Ua$ we can utilize \REFprop{prop:Realize} to extract a maximal model $M$ of $\Autap$ if it exists. We now show how we can define an output-feedback control strategy $\Ca$ from $M$.

\begin{definition}\label{def:MtoCa}
 Let $\Autap$ be the UCA induced by $\Sa$, $\psi$ and $\Mapa{}$ and $M=(Z,\set{a_0},\Ya\times\Ua,\alpha)$ a maximal model of $\Autap$.
 Then we say that $\Ca:\EPrefs{\Sa}\fun\Ua$ is an output-feedback control strategy induced by $M$ if for all $\nu=\ya_0\ua_0\ya_1\hdots\ya_k\in\EPrefs{\Sa}$ we have that (i) $\C(\nu)=\ua_k$ implies the existence of $z:=\LastSn{M}{\ya_0\ua_0\ya_1\hdots\ya_{k-1}\ua_{k-1}}$ s.t.\ $|\alpha(z,(\ya_k,\ua_k))|=1$ and (ii) $\Ca(\nu)=\emptyset$ only if $\nu\notin\EPrefs{M}$.
\end{definition}

The next theorem, which is the main result of this section, shows that the construction of an output-feedback controller $\Ca$ for a finite system $\Sa$ with specification $(\psi,\Mapa{})$ via \REFprop{prop:Auts}, \REFdef{def:synthUCA}, \REFprop{prop:Realize} and \REFdef{def:MtoCa}, is sound and complete. I.e., if the controller $\Ca$ obtained from $M$ is non-empty, it is sound (i.e., $\Ca\in\WIN(\Sa,\psi,\Mapa{})$), and if no non-empty controller can be derived from $M$ no solution to the given control problem exists (i.e. $\WIN(\Sa,\psi,\Mapa{})=\emptyset$).

\begin{theorem}\label{thm:Casound_extended}
Let $\Autap$ be the UCA induced by $\Sa$, $\psi$ and $\Mapa{}$, and $M$ be a maximal model $M$ of $\Autap$. Then (i) we have $\Ca\in\WIN(\Sa,\psi,\Mapa{})$ for every \emph{non-empty} controller $\Ca$ induced by $M$, and (ii) if no non-empty controller induced by $M$ exists, we have $\WIN(\Sa,\psi,\Mapa{})=\emptyset$.
\end{theorem}

The proof of \REFthm{thm:Casound_extended} is given in \REFapp{sec:app:Casound}. For soundness (\REFthm{thm:Casound_extended} (i)) the claim intuitively follows from the observation that \REFprop{prop:synthUCA_prop} (ii) ensures that $\Ca$ is feedback-composeable with $\Sa$ and (i) ensures that all paths of $(\Sa,\Ca)$ are compatible with $\psi$. 

Conversely, completeness (\REFthm{thm:Casound_extended} (ii)) is established by showing that whenever $\WIN(\Sa,\psi,\Mapa{})\neq\emptyset$ there exists a non-empty controller $\Ca$ induced by $M$. Intuitively, this follows from the observations that 
(a) $\WIN(\Sa,\psi,\Mapa{})\neq\emptyset$ implies the existence of some infinite path $\sigma\in\EPaths{\Sa}$ s.t.\ $\Mapa{\Extn{\Sa}^{-1}(\sigma)}\subseteq\semantics{\psi}$, which in turn implies (from \REFprop{prop:synthUCA_prop} (i)) that (b) $\sigma\in\Lang(\Autap)$ and therefore (c) $M$ is non-empty (as $M$ is defined to be maximal), implying the existence of a non-empty $\Ca$ by definition.

\section{Related Work}\label{sec:RelatedWork}

Our paper is closest related to \cite{ReissigWeberRungger_2017_FRR,MajumdarNS2020,khaled2020outputfeedback,MizoguchiUshio2018deadlockfree}.  

Our definition of sound abstractions and sound abstract specification is based on the definition of the same notions for state-feedback control in \cite{ReissigWeberRungger_2017_FRR}. Indeed, targeting all maps used in this paper to the special case discussed in \cite{ReissigWeberRungger_2017_FRR} we observe that both definitions coincide.
In particular, we recover the situation discussed in \cite{ReissigWeberRungger_2017_FRR} if $\I=U$, $\O=X$, $\Ia=\Ua$ and $\Oa=\Xa$, and therefore $\Imap{}=\Omap{}=\Imapa{}=\Omapa{}=\iota$. Further, we choose $\beta$ to be an identity map that is strictly defined on $\Ua$ and  have $Y=X$ and $H=\iota$. Therefore $\gamma$ coincides with $\alpha$.  With this, \REFdef{def:SoundAbs} coincides with the definition of FRR in \cite[Def. V.2]{ReissigWeberRungger_2017_FRR} and \REFdef{def:soundaspec} coincides with \cite[Def. VI.2]{ReissigWeberRungger_2017_FRR}.

In \cite{MajumdarNS2020} the authors consider output-feedback control design with specifications defined over a finite set of observables. This corresponds to the setting in this paper when defining $\O=Y=\Oa=\Ya$ and $\I=U=\Ia=\Ua$ implying that $\Omap{}=H$, $\Omapa{}=\Ha$ and $\Imap{},\Imapa{}$ are identity maps. Further, as $Y=\Ya$ and $U=\Ua$ the maps $\beta$ and $\gamma$ are also the identity maps. With this \REFdef{def:SoundAbs} coincides with \cite[Def.3.1]{MajumdarNS2020}\footnote{We remark that there was a typo in \cite[Def. 3.1]{MajumdarNS2020}. The set $\Ya$ in the definition of $\Sa$ must be $Y$.}. While the abstraction algorithm in \cite{MajumdarNS2020} does not require a grid-based discretization of the state space, it does not allow to handle state-based specifications.

Recently, Khaled et.al. proposed a similar notion of sound abstractions for symbolic output-feedback control in \cite{khaled2020outputfeedback}. Unfortunately, it seems to the authors that their definition of sound abstractions does not allow to prove a soundness result similar to \REFthm{thm:ABoCD}. Further, the authors only consider the abstract output-feedback controller synthesis problem either for specifications over the observables, similar to \cite{MajumdarNS2020}, or for abstractions that are detectable, i.e., where the state of the system becomes observable after a finite number of observations. Our controller synthesis procedure does not require these assumptions.  

Finally, \cite{MizoguchiUshio2018deadlockfree} solves the outlined ABoCD problem for safety specifications only, while we can handle arbitrary LTL properties.

\bibliographystyle{abbrv}
{%\tiny
\bibliography{reportbib}}

\newpage 
\appendix
\subsection{Proofs of \REFlem{lem:ABoCD}}\label{sec:app:ABoCD}

We prove all claims by induction over $k$.

 \begin{inparaitem}[$\blacktriangleright$]
   \item Base case ($k=0)$:\\
   \begin{inparaitem}[$\triangleright$]
    \item \textbf{(a)} We have $\sigma_0=y_0\in\EPrefs{S}$ iff $y_0\in H(\Xo)$. As $\EPrefsk{S,\C}{k}=H(\Xo)$ by definition, the claim follows.\\
%     \item (b)
%     We have $x_k\in\Extn{S}^{-1}(\sigma_{0})=\Extn{S}^{-1}(y_{0})=H^{-1}(y_0)\subseteq X_0$
%     $\pi_0=x_0\in\CPrefs{S}$ iff $x_0\in \Xo$. As  $\sigma_0=y_0\in H(\Xo)$, we have $\pi_0=x_0\in\Extn{S}^{-1}(y_0)$ iff $x_0\in \Xo$ by definition and therefore $x_0\in\CPrefsk{S,\C}{0}$ iff $x_0\in\CPrefs{S}$.\\
    \item \textbf{(b/c)} Observe that $\ya_0\in\Omega_{\beta,\gamma}(y_0)$ implies $\ya_0\in\gamma(y_0)$. Further, $x_0\in\Extn{S}^{-1}(y_0)$ implies $x_0\in\Xo$ and $y_0\in H(x_0)$. It then follows from (A1) that $\emptyset\neq\alpha(x_0)\subseteq \Xao$. Now observe that all $\alpha(x_0)=\Omega_{\alpha,\beta}(x_0)$ and it follows from  (A3) that for all $\xa_0\in\alpha(x_0)$ holds that $\ya_0\in\Ha(\xa_0)$ implying $\xa_0\in\Extn{\Sa}^{-1}(\ya_{0})$, proving (c).     
    As $\EPrefsk{\Sa,\Ca}{0}:=\Ha(\Xao)$ by definition, we further conclude $\ya_0\in\EPrefsk{\Sa,\Ca}{0}$, proving (b). \\
    \item \textbf{(d)} First, recall that $\Ca$ is feedback composable with $\Sa$. Therefore, $\ya_0\in\EPrefsk{\Sa,\Ca}{0}$ implies that $\Ca(\ya_0)\in\Enab(\LastSn{\Sa}{\ya_0})\neq\emptyset$. Further, it follows from (b) that for all $x_0\in\Extn{S}^{-1}(y_{0})$ holds that $\alpha(x_0)\subseteq\LastSn{\Sa}{\ya_{k}}$. This implies $\ua_{0}:=\Ca(\ya_0)\in\Enab(\alpha(x_{k}))$. Then it follows from (A2.i) that for any $u_{0}\in\beta(\ua_0)$ holds that $u_{0}\in\Enab(x_{0})$ and hence  
    $u_{0}\in\Enab_{S}(\LastSn{S}{y_0})$. With this, the claim follows from the constuction of $\C$.\\
   \end{inparaitem}

   \item Induction step from $k-1$ to $k$: 
   It follows from induction hypothesis (d)  that $\emptyset\neq\C(\sigma_{k-1})\subseteq\Enab_{S}(\LastSn{S}{\sigma_{k-1}})$ for any choice of $\sigma_{k-1}$. We consider any coice $u_{k-1}\in\C(\sigma_{k-1})$.\\
\begin{inparaitem}[$\triangleright$]
    \item \textbf{(a)} follows from the definition of $\EPrefsk{S,\C}{k}$. \\
    \item \textbf{(b/c)} We fix $\pi_k=x_0u_0\hdots x_k$ and $\sigmaa_k=\ya_0\ua_0\hdots\ya_{k}$ as required. Then it follows from induction hypothesis (c) that $\pia_{k-1}\in\Extn{\Sa}^{-1}(\sigmaa_{k-1})$ for any $\pia_{k-1}$ compatible with $\pi_{k-1}=\pi_k|_{[0;k-1]}$ and $\sigmaa_{k-1}=\sigmaa_k|_{[0;k-1]}$. This implies $\alpha(x_{k-1})\subseteq\LastSn{\Sa}{\sigmaa_{k-1}}$ by construction. 
    
    As $\sigmaa_{k-1}\in\EPrefsk{\Sa,\Ca}{k-1}$ from induction hypothesis (b) and $\Ca$ is a sound controller for $\Sa$ we know that $\ua_{k-1}\in\Enab(\LastSn{\Sa}{\sigmaa_{k-1}})\neq\emptyset$. This implies $\ua_{k-1}\in\Enab(\alpha(x_{k-1}))$. Then it follows from (A2.i) that $u_{k-1}\in\Enab(x_{k-1})$ and from (A2.ii) that $\alpha(F(x_{k-1},u_{k-1}))\subseteq\Fa(\xa_{k-1},\ua_{k-1})$ for all $\xa_{k-1}\in\alpha(x_{k-1})$.
    Now recall that $x_k\in F(x_{k-1},u_{k-1})$ which implies $\xa_{k}\in\Fa(\xa_{k-1},\ua_{k-1})$ for any $\xa_{k-1}\in\alpha(x_{k-1})$ and any $\xa_{k}\in\alpha(x_{k})$. This implies by construction that $\pia_k\in\Extn{\Sa}^{-1}(\sigmaa_{k})$ for any considered $\pia_k$, which proves claim (b).    
    
    Now, recall from the construction of $\sigmaa_k$ in (c) that $\ya_k\in\gamma(H(x_k))$. Then it follows from (A3) that $\ya_k\in\Ha(\xa_k)$ for any $\xa_{k}\in\alpha(x_{k})$. With this, claim (c) follows from the definition of $\EPrefsk{\Sa,\Ca}{k}$.\\
    \item \textbf{(d)} It follows from (b) that there exists a $\sigmaa_k\in\Omega_{\beta,\gamma}(\sigma_{k})$ s.t.\ $\sigmaa_{k}\in\EPrefsk{\Sa,\Ca}{k}$. As $\Ca$ is a sound controller for $\Sa$ it follows that for all such $\sigmaa_k$ we have $\Ca(\sigmaa_k)\neq\emptyset$ and $\Ca(\sigmaa_k)\in\Enab(\LastSn{\Sa}{\sigmaa_{k}})$. Now it follows form (c) that for any $\pi_k=x_0u_0\hdots x_k\in\Extn{S}^{-1}(\sigma_{k})$ holds that $\alpha(x_k)\subseteq\LastSn{\Sa}{\sigmaa_{k}}$. This implies $\ua_{k}:=\Ca(\sigmaa_k)\in\Enab(\alpha(x_{k}))$. Then it follows from (A2.i) that for any $u_{k}\in\beta(\ua_k)$ holds that $u_{k}\in\Enab(x_{k})$ and hence  
    $u_{k}\in\Enab_{S}(\LastSn{S}{\sigma_k})$. With this, the claim follows from the constuction of $\C$ in \eqref{equ:ABoCD}.
    \end{inparaitem}
   \end{inparaitem}

\subsection{Additional Proofs for \REFprop{prop:synthUCA_prop}}\label{sec:app:synthUCA}

In this section we formalize the intuition behind \REFprop{prop:synthUCA_prop}. In particular, we prove a stronger claim, formalized in the following 
\REFprop{prop:synthUCA_prop_extended}, from which \REFprop{prop:synthUCA_prop} follows as a corollary.

\begin{proposition}\label{prop:synthUCA_prop_extended}
Given the premises of \REFdef{def:synthUCA}, let $\sigma=y_0u_0y_1u_1\hdots\in Y(UY)^\omega$ be an external sequence. Then 
$\sigma\in\Lang(\Autap)$ iff $\sigma\notin\ON{IBlock}(\Sa)$ and either 
\begin{compactenum}[(i)]
\item $\sigma\notin\EPaths{\Sa}$, or 
\item $\sigma\in\EPaths{\Sa}$ and $\Mapa{\Extn{\Sa}^{-1}(\sigma)}\subseteq\semantics{\psi}$.
\end{compactenum}
\end{proposition}

% In order to prove \REFprop{prop:synthUCA_prop_extended}, we first formally define the set $\ON{IBlock}(\Sa)$. 
% 
% \begin{definition}\label{def:IBlock}
% Given a system $\Sa$ we define its set of blocking sequences $\ON{IBlock}(\Sa)\subseteq \Ya(\Ua\Ya)^\omega$ as follows. An external sequence $\sigma=y_0u_0y_1u_1\hdots\in Y(UY)^\omega$ is said to be blocking on $\Sa$, i.e., $\sigma\in\ON{IBlock}(\Sa)$, iff there exists an index $k\in\mathbb{N}$ s.t.\ $\nu=\sigma|_{[0,k]}\in\EPrefs{\Sa}$ but $u_k\notin\Enab(\LastSn{\Sa}{\nu})$.
% \end{definition}
% 
% We emphasize that the map $\Enab$ is lifted to sets of states using \emph{intersection}. Hence, a string $\sigma$ is called \emph{blocking on $\Sa$} as defined in \REFdef{def:IBlock}, if there exists some state $q\
% in\LastSn{\Sa}{\nu}$ in which $u_k$ is not enabled, even if there might exist other states $q'\in\LastSn{\Sa}{\nu}$ which allow progress on $u_k$. We will come back to this point for the proof of \REFthm{thm:Casound}.

% Recalling that a UCA accepts a trace if its resulting run over $\Aut$ is finite, we see that we have constructed a UCA in \REFdef{} that only \emph{rejects} sequences that block on not enabled inputs. We collect these sequences in a set $\ON{IBlock}(\Sa)\subseteq \Ya(\Ua\Ya)^\omega$. Formally, an external infinite sequence $\sigma=y_0u_0y_1u_1\hdots$ belongs to $\ON{IBlock}(\Sa)$ iff there exists an index $k\in\mathbb{N}$ and a prefix $\nu=\sigma|_{[0,k]}$ s.t.\ $\LastSn{\Sa}{\nu}\neq\emptyset$ and $u_k\notin\Enab(\LastSn{\Sa}{\nu})$.
% With this, we have the following obvious proposition.

We now prove \REFprop{prop:synthUCA_prop_extended} in multiple steps. We first show that an external sequence $\sigma$ has a compliant \emph{infinite} path over $\Autap$ iff either $\sigma\in\EPaths{\Sa}$ or $\sigma\in\ON{IBlock}(\Sa)$. We further show that these two cases are disjointed and that we know that $\sigma\in\ON{IBlock}(\Sa)$ iff there exists a compliant \emph{infinite} path over $\Autap$ that visits $\bot$ infinitely often. This is formalized in the following lemma.

\begin{lemma}\label{lem:synthUCA_finite}
Given the premises of \REFdef{def:synthUCA}, let $\sigmaa\in \Ya(\Ua\Ya)^\omega$. Then there exists an \emph{infinite} path $\pia$ compliant with $\sigmaa$ iff $\sigmaa\in\EPaths{\Sa}\dot{\cup}\ON{IBlock}(\Sa)$. 
Further,  there exists a path $\pia$ compliant with $\sigmaa$ that visits $\bot$ infinitely often iff $\sigmaa\in\ON{IBlock}(\Sa)$.
\end{lemma}

\begin{proof}
\enquote{$\Leftarrow$}:
 First recall that $\Auts$ is complete in $\O\times\I$. With this, it follows from the definition of $\Autap$ that for any path of $\Sa$ there exists a path over $\Autap$ with the same length. With this, it immediately follows that whenever $\sigmaa\in\EPaths{\Sa}$, there is by definition an infinite path over $\Sa$ compliant with $\sigmaa$ and therefore this implies that there is an infinite compliant path over $\Autap$ as well. 
 
 Further, if $\sigmaa\in\ON{IBlock}(\Sa)$ instead, we know that there exists a prefix $\nu=\sigmaa|_{[0;k]}$ s.t.\ $\LastSn{\Sa}{\nu}\neq\emptyset$ and $\ua_k\notin\Enab(\LastSn{\Sa}{\nu})$. Then it again follows from the completeness of $\Auts$ that there exists a path over $\Autap$ that is compliant with $\nu$ and for which 
 $\ua_k\notin\Enab(\LastSn{\Sa}{\nu})$. This implies from the construction that $\qa_k=\bot$. As $\bot$ is absorbing this implies that there is an infinite compliant path for $\sigmaa$ that visits $\bot$ infinitely often.
 
 \enquote{$\Rightarrow$}: Let $\pia$ be an infinite path compliant with $\sigmaa$. This implies that either (i) there exists some $k\in\N$ s.t.\ $\qa_{k}=(\xa_k,q_k)\neq\bot$ and $\qa_{k'}=\bot$ for all $k'>k$ or (ii) $\xa_{k+1}\in\Fa(\xa_k,\ua_k)$ and $\ya_k\in\Ha(\xa_k)$ for all $k\in\N$. If (i) holds, we see that $\pia$ visits $\bot$ infinitely often. We also see that, by construction,$\Fa(\xa_k,\ua_k)=\emptyset$ and hence $\ua_k\notin \Enab(\xa_k)$. With this it follows from the lifting of $\Enab$ to sets via \emph{intersection} that $\ua_k\notin\Enab(\LastSn{\Sa}{\sigma|_{[0,k]}})$ while $\xa_k\in\LastSn{\Sa}{\sigmaa|_{[0,k]}}\neq\emptyset$. Therefore, $\sigmaa\in\ON{IBlock}(\Sa)$. If (ii) holds it immediately follows from the definition that $\pia\in\Paths{\Sa}$ and hence, $\sigmaa\in\EPaths{\Sa}$.
\end{proof}

After establishing \REFlem{lem:synthUCA_finite}, we see that all external sequences that only have \emph{finite} compliant paths over $\Autap$ must not be contained in $\EPaths{\Sa}$ and obviously also not in $\ON{IBlock}(\Sa)$. With this, it remains to show that for all sequences $\sigmaa\in\EPaths{\Sa}$ with a compliant infinite play over $\Autap$ holds that \emph{all} compliant sequences over $\Autap$ fulfill the specification. This is formalized in the following lemma.

\begin{lemma}\label{lem:synthUCA_spec}
Given the premises of \REFdef{def:synthUCA}, let $\sigmaa\in\EPaths{\Sa}$. Then $\sigmaa\in\Lang(\Autap)$ iff $\Mapa{\Extn{\Sa}^{-1}(\sigmaa)}\subseteq\semantics{\psi}$.
\end{lemma}

\begin{proof}
First observe that the definition of $\Extn{\Sa}$ ensures that $\Extn{\Sa}^{-1}(\sigmaa)$ only contains \emph{infinite} paths $\nu=\xa_0\ua_0\xa_1\ua_1\hdots$ of $\Sa$ s.t.\ for all $k\in\N$ we have $\ya_k\in\Ha(\xa_k)$ and $\ua_k\in\Fa(\xa_k,\ua_k)$. As $\Auts$ is complete on $\O\times\I$ we know that for all $\rho\in\Mapa{\nu}$ (for any $\nu\in\Extn{\Sa}^{-1}(\sigmaa)$, i.e., for all $\rho\in\Mapa{\Extn{\Sa}^{-1}(\sigmaa)}$) there exists an \emph{infinite} path $\alpha=q_0(\o_0,\i_0)q_1\hdots$ over $\Auts$. Now it follows from the construction of $\Autap$ that for any such $\nu$ and $\alpha$ there is an infinite path $\pi=(\xa_0,q_0)(\ua_0,\ya_0)(\xa_1,q_1)\hdots$ over $\Autap$. We see that any such path $\pi$ is compliant with $\sigmaa$ and never visits $\bot$. We therefore know that $\pi$ visits $\acca$ infinitely often iff $\alpha$ visits $\acc$ infinitely often. With this we have $\sigmaa\in\Lang(\Autap)$ iff $\rho\in\Lang(\Auts)=\semantics{\psi}$. As this reasoning holds for all $\rho\in\Mapa{\Extn{\Sa}^{-1}(\sigmaa)}$ the claim follows.
\end{proof}

With this, \REFprop{prop:synthUCA_prop_extended} becomes a direct consequence of \REFlem{lem:synthUCA_finite} and \REFlem{lem:synthUCA_spec} as formalized in the final proof below.

% \begin{proposition}
% Given the premisses of \REFdef{}, let $\sigma=y_0u_0y_1u_1\hdots\in Y(UY)^\omega$ be an external sequence. Then 
% $\sigma\in\Lang(\Autap)$ iff either 
% \begin{compactenum}[(i)]
%  \item $\sigma\notin\EPaths{\Sa}\cup\ON{IBlock}(\Sa)$, or
%  \item $\sigma\in\EPaths{\Sa}$ and $\Mapa{\Extn{\Sa}^{-1}(\sigma)}\subseteq\semantics{\psi}$.c
% \end{compactenum}
% \end{proposition}
% 
\begin{proof}
  We proof both directions separately.\\
  \begin{inparaitem}[$\blacktriangleright$]
    \item \enquote{$\subseteq$}:
    Let $\sigmaa=\ya_0\ua_0\ya_1\ua_1\hdots\in\Lang(\Autap)$. This implies that all paths $\pia=\qa_0(\ya_0,\ua_0)\qa_1(\ya_1,\ua_1)\qa_2\hdots$ over $\Autap$ that are compliant with $\sigmaa$ are either finite or are infinite and visit $\acca$ only finitely often.
    
    Now it follows from \REFlem{lem:synthUCA_finite} that whenever \emph{all} compliant paths are finite, we know that $\sigmaa\notin\EPaths{\Sa}$ and $\sigmaa\notin\ON{IBlock}(\Sa)$ and $\EPaths{\Sa}$ and $\ON{IBlock}(\Sa)$ are disjoint. We now consider the case that there exists an infinite compliant path. As this path only visits $\acca$ finitely often it follows from \REFlem{lem:synthUCA_finite} that we again have $\sigmaa\notin\ON{IBlock}(\Sa)$.
    We therefore know that $\sigmaa\notin\EPaths{\Sa}$. With this, it follows from \REFlem{lem:synthUCA_spec} and the fact that $\sigmaa\in\Lang(\Autap)$ that we also have $\Mapa{\Extn{\Sa}^{-1}(\sigmaa)}\subseteq\semantics{\psi}$.\\
    \item \enquote{$\supseteq$}: We have $\sigmaa\notin\ON{IBlock}(\Sa)$ and consider two cases. If $\sigmaa\notin\EPaths{\Sa}$ it follows from \REFlem{lem:synthUCA_finite} that all paths over $\Autap$ that are compliant with $\sigmaa$ are finite, and hence $\sigmaa\in\Lang(\Autap)$ from the definition of acceptance of UCA. Now let $\sigmaa\in\EPaths{\Sa}$ and $\Mapa{\Extn{\Sa}^{-1}(\sigmaa)}\subseteq\semantics{\psi}$. Then it follows from \REFlem{lem:synthUCA_spec} that $\sigmaa\in\Lang(\Autap)$.
    \end{inparaitem}
\end{proof}

\subsection{Additional Proofs for \REFthm{thm:Casound_extended}}\label{sec:app:Casound}

% In this section we formalize the intuition behind \REFthm{thm:Casound}. In particular, we prove a stronger claim, formalized in the following 
% \REFthm{thm:Casound_extended}, from which \REFthm{thm:Casound} follows as a corollary.
% 
% 
% \begin{theorem}\label{thm:Casound_extended}
% Let $\Autap$ be the UCA induced by $\Sa$, $\psi$ and $\Mapa{}$, and $M$ be a non-empty maximal model $M$ of $\Autap$. Then (i) we have $\Ca\in\WIN(\Sa,\psi,\Mapa{})$ for every \emph{non-empty} controller $\Ca$ induced by $M$, and (ii) if no non-empty controller induced by $M$ exists, we have $\WIN(\Sa,\psi,\Mapa{})=\emptyset$.
% \end{theorem}

We prove part (i) and part (ii) of \REFthm{thm:Casound_extended} separately.

\smallskip
\noindent\textbf{\REFthm{thm:Casound_extended}, Part (i).}
This proof reduces to showing that $\Ca$ is feedback-composable with $\Sa$ and for all $\sigmaa\in\EPaths{S,\C}$ holds $\Mapa{\Extn{\Sa}^{-1}(\sigmaa)}\subseteq\semantics{\psi}$. We show this in two steps using \REFlem{lem:M_i} and \REFlem{lem:Ca_i} below.

\begin{lemma}\label{lem:M_i}
  Given the premises of \REFthm{thm:Casound_extended} it holds that $\sigmaa\in\EPaths{\Sa}$ and $\sigmaa\in\EPaths{M}$ implies $\Mapa{\Extn{\Sa}^{-1}(\sigmaa)}\subseteq\semantics{\psi}$.
\end{lemma}

\begin{proof}
First, it follows from the definition of \emph{maximal models} that $\sigmaa\in\Lang(\Autap)$. With this, it follows from \REFlem{lem:synthUCA_spec} that $\Mapa{\Extn{\Sa}^{-1}(\sigmaa)}\subseteq\semantics{\psi}$.
\end{proof}

\begin{lemma}\label{lem:Ca_i}
 Given the premises of \REFthm{thm:Casound_extended} let $\Ca$ be a non-empty controller induced by $M$. Then $\Ca$ is feedback-composable with $\Sa$ and for all $\sigmaa\in\EPaths{S,\C}$ holds $\Mapa{\Extn{\Sa}^{-1}(\sigmaa)}\subseteq\semantics{\psi}$.
\end{lemma}

\begin{proof}
As $\Ca$ is non-empty, we know that there exists at least one $\sigmaa\in\EPaths{\Sa}$ s.t.\ $\sigmaa\in\EPaths{M}$ which in turn implies $\sigmaa\in\Lang(\Autap)$, and hence $\sigmaa\notin\ON{IBlock}(\Sa)$ (from \REFlem{lem:synthUCA_finite}).
It further follows from the construction of $\Ca$ in \REFdef{def:MtoCa} that there exists at least one such $\sigmaa$ s.t.\ for all $k\in\N$ we have $\Ca(\sigmaa|_{[0;k]})=\ua_k$. As we have $\sigmaa\notin\ON{IBlock}(\Sa)$ also for this $\sigmaa$, it follows that $\ua_k\in\Enab(\LastSn{\Sa}{\ya_0\ua_0\hdots\ya_k})$ for every $k$, which implies that $\Ca$ is feedback-composable with $\Sa$ and that $\sigmaa\in\EPaths{S,\C}$. It now follows from \REFlem{lem:M_i} and the observation that $\sigmaa\in\EPaths{\Sa}$ and $\sigmaa\in\EPaths{M}$ that $\Mapa{\Extn{\Sa}^{-1}(\sigmaa)}\subseteq\semantics{\psi}$.
\end{proof}

It is now easy to see that \REFthm{thm:Casound_extended} (i) follows directly from \REFlem{lem:Ca_i} and the definition of $\WIN(\Sa,\psi,\Mapa{})$.

\smallskip
\noindent\textbf{\REFthm{thm:Casound_extended}, Part (ii).}
We prove this claim by showing that whenever $\WIN(\Sa,\psi,\Mapa{})\neq\emptyset$ we know that there exists a non-empty output-feedback control strategy induced by $M$. We prove this claim in multiple steps.

\begin{lemma}\label{lem:CatoM_1}
 Let $\WIN(\Sa,\psi,\Mapa{})\neq\emptyset$. Then there exists a $\sigmaa\in\EPaths{\Sa}$ s.t.\ $\Mapa{\Extn{\Sa}^{-1}(\sigmaa)}\subseteq\semantics{\psi}$.
\end{lemma}

\begin{proof}
 If $\WIN(\Sa,\psi,\Mapa{})\neq\emptyset$ we know that there exists a control strategy $\Ca$ which is feedback composable with $\Sa$ and $\Map{\CPaths{\Sa, \Ca}}\subseteq\semantics{\psi}$. As $\semantics{\psi}\subseteq \O(\I\O)^\omega$ we know that $\CPaths{\Sa, \Ca}$ can only contain infinite paths. Now pick such an infinite $\pi\in\CPaths{\Sa, \Ca}$ and consider some $\sigmaa\in\Extn{\Sa}(\pi)$. Then we know from the definition of $\CPaths{\Sa, \Ca}$ that $\sigmaa\in\EPaths{\Sa, \Ca}$ and hence $\Extn{\Sa}^{-1}(\sigmaa)\subseteq\CPaths{\Sa, \Ca}$. With this, it follows that $\Mapa{\Extn{\Sa}^{-1}(\sigmaa)}\subseteq\semantics{\psi}$ which proves the claim.
\end{proof}

\begin{lemma}\label{lem:CatoM_2}
 If $\WIN(\Sa,\psi,\Mapa{})\neq\emptyset$ then $\EPaths{\Sa}\cap\Lang(\Autap)\neq \emptyset$.
\end{lemma}

\begin{proof}
 We know from \REFlem{lem:CatoM_1} that there exists at least one $\sigmaa\in\EPaths{\Sa}$ s.t.\ $\Mapa{\Extn{\Sa}^{-1}(\sigmaa)}\subseteq\semantics{\psi}$. Then it follows from \REFlem{lem:synthUCA_spec} that $\sigmaa\in\Lang(\Autap)$. This proves the claim.
\end{proof}

\begin{lemma}\label{lem:CatoM_3}
 If $\EPaths{\Sa}\cap\Lang(\Autap)\neq \emptyset$ then there exists a non-empty output-feedback control strategy induced by $M$.
\end{lemma}

\begin{proof}
As $\EPaths{\Sa}\cap\Lang(\Autap)\neq \emptyset$ we know that there exists a $\sigmaa\in\EPaths{\Sa}\cap\Lang(\Autap)$. With this, it follows from \REFlem{lem:synthUCA_finite} that there exists an infinite path $\pi$ over $\Autap$ that is compliant with $\sigmaa$. This implies that $\pi$ is also a path of $M$ and hence, $\sigmaa\in\EPaths{M}$. As we know that $\sigmaa\in\EPaths{\Sa}$ from above, we see that there exists a non-empty control strategy induced by $M$ from \REFdef{def:MtoCa}.
\end{proof}

With this, \REFthm{thm:Casound_extended}, Part (ii) becomes a direct consequence of \REFlem{lem:CatoM_2} and \REFlem{lem:CatoM_3}.

\end{document}